\documentclass[11pt]{article}

\usepackage[margin=1in]{geometry}
\usepackage{amsmath,amssymb,amsthm,mathtools}
\usepackage{graphicx}
\usepackage{booktabs}
\usepackage{array}
\usepackage{enumitem}
\usepackage{hyperref}
\usepackage{xcolor}
\usepackage{caption}
\usepackage{subcaption}
\usepackage{float}

\hypersetup{colorlinks=true, linkcolor=blue!60!black, citecolor=blue!60!black, urlcolor=blue!60!black}

\newtheorem{proposition}{Proposition}
\newtheorem{remark}{Remark}
\newtheorem{assumption}{Assumption}
\newcommand{\R}{\mathbb{R}}
\newcommand{\E}{\mathbb{E}}
\newcommand{\dd}{\mathrm{d}}

\title{Quality-Adjusted Hit-Ratio Targeting in Corporate Bond Market Making\\

\large Credit Alpha, Client Flow Quality, and Style-Aware Warehousing}

\author{Bouna NIANG\thanks{The author publishes this paper in a personal capacity as an independent quantitative trader/researcher. The views expressed are solely those of the author and do not represent the views of any former employer, including HSBC and Citi, or any of their affiliates. All errors are the author's own. This paper builds on the cited academic literature and public sources; all modelling choices, simulations, interpretations, and conclusions are the author's. The author would like to thank Casimir Emako Kazianou for fruitful discussions on HJB equations and numerical methods.}}
\date{\today}

\begin{document}
\maketitle

\begin{abstract}
Hit ratio is a common service metric for electronic corporate bond market making, but raw hit-ratio targets can be economically misleading when client flow has heterogeneous adverse-selection content. This paper extends a stochastic-control framework for OTC bond RFQ market making with hit-ratio constraints by replacing raw hit ratio with a residual-quality-adjusted hit ratio. The key modelling distinction is that adverse post-trade markouts are first decomposed into observable credit factors, carry/rolldown, issuer-relative-value effects, index or ETF demand effects, and residual adverse selection. Only the residual component is treated as client-flow toxicity. The resulting control problem remains tractable: after dualizing the quality-hit-ratio penalty, the HJB retains separable Hamiltonians, and the dual variable is the solution of an exact one-dimensional nonlinear fixed point for each targeted tier. Under a quadratic value-function approximation, optimal quotes decompose into a riskless spread, inventory skew, credit-alpha skew, residual-toxicity charge, and quality-hit-ratio subsidy. Synthetic multi-bond simulations with nonlinear dual solves illustrate that raw hit-ratio targeting can subsidize residual-toxic flow, while residual-quality targeting reallocates service toward low-residual-toxicity flow and improves the attained service/economics frontier. A final reduced-form extension studies inventory-recycling value through risk-aware style-aligned client-flow warehousing. Sweep or portfolio-trade opportunities fill randomly, and participation is sized using the same quadratic value approximation as the RFQ quoting problem. A passive/index-demand experiment is reported in the appendix as a special case of forecastable client flow. The numerical evidence is synthetic and mechanism-oriented; no proprietary RFQ data are used.
\end{abstract}

\paragraph{Scope.}
The contribution is a tractable stochastic-control extension and a set of controlled synthetic simulations. The numerical experiments are not empirical performance claims: no proprietary RFQ, client, or markout data are used. Their purpose is to demonstrate the economic mechanisms of residual-quality targeting and to provide a reproducible research design for later calibration.

\section{Introduction}

Electronic corporate bond market making is not only an inventory-risk problem. Dealers also face client-service targets, relationship constraints, and platform-level performance metrics. A natural metric is hit ratio: the fraction of RFQs won, often size-weighted and segmented by client tier. Recent work by Barzykin and Ciceri \cite{barzykin_ciceri_2026} shows how to include a running hit-ratio target in an OTC bond market-making HJB while preserving tractability by dualizing the target penalty. Their framework is an attractive base because it produces exact quote maps and a scalable quadratic approximation.

This paper argues that, in credit, raw hit ratio is not the correct service objective. A dealer should not want to win all flow equally. Some RFQs are low residual-adverse-selection, help recycle inventory, or come from clients trading on transparent carry, rolldown, index, or issuer-relative-value considerations. Other RFQs have persistently adverse residual markouts even after controlling for observable credit factors. A uniform raw hit-ratio target can force the dealer to subsidize the latter flow.

The central proposal is to replace raw hit ratio with a \emph{quality-adjusted hit ratio}. Each RFQ bucket receives a quality weight based on residual toxicity, where residual toxicity is defined only after removing observable credit alpha and factor effects from signed markouts. A client buying high-carry or cheap issuer-curve bonds should not be labelled toxic simply because the subsequent markout is adverse to the dealer; that effect should be attributed to the public or structural credit signal that the dealer failed to incorporate. The toxicity term should capture the remaining, client-flow-specific adverse-selection component.

The contribution is fourfold.
\begin{enumerate}[leftmargin=*]
    \item We define residual client-flow toxicity through factor- and alpha-adjusted post-trade markouts.
    \item We introduce a residual-quality-adjusted hit ratio and show that the dualized HJB remains separable.
    \item We derive a quote decomposition under a quadratic/Riccati approximation: riskless spread, inventory skew, credit-alpha skew, residual-toxicity charge, and quality-hit-ratio subsidy.
    \item We use synthetic multi-bond credit simulations to demonstrate the service/economics frontier: conditional on delivering a service target, residual-quality targeting can dominate raw hit-ratio targeting.
\end{enumerate}

A reduced-form extension studies inventory-recycling value through style-aware client-flow warehousing: the dealer can selectively participate in client inventory opportunities when the offered risk is aligned with the desk's credit-style alpha, but participation must be sized for random fills and inventory curvature. A passive/index-flow experiment is retained in the appendix as a special case of forecastable low-residual-toxicity flow.

\subsection{Scope and main claim}

The main claim is deliberately narrow. Hit ratio is not intrinsically a bad objective; a uniform raw hit-ratio target is bad when the economic quality of fills is heterogeneous. The paper therefore studies a conditional problem: given that a dealer wishes to deliver a client-service target, what target should be used? The answer proposed here is a residual-quality-adjusted hit ratio rather than a raw hit ratio.

The synthetic experiments are designed to test the sign and mechanism of the result. They are not intended to estimate the profitability of any particular desk. An empirical implementation would require calibration of arrival intensities, fill curves, credit factor exposures, alpha signals, and factor-adjusted markouts from RFQ data.

\subsection{Relation to the literature}

The model sits in the stochastic-control market-making tradition of Avellaneda and Stoikov \cite{avellaneda_stoikov_2008} and Gueant, Lehalle, and Fernandez-Tapia \cite{gueant_lehalle_tapia_2013}. The closest reference is Barzykin and Ciceri \cite{barzykin_ciceri_2026}, who introduce a running hit-ratio target in an OTC bond RFQ market-making model and show that dualizing the target preserves separable Hamiltonians and a scalable Riccati approximation. The present paper keeps that tractability but changes the service objective from raw hit ratio to residual-quality-adjusted hit ratio.

\section{Baseline RFQ market-making model}

Consider $d$ corporate bonds indexed by $m\in\{1,\ldots,d\}$. Let $Q_t\in\R^d$ denote dealer inventory in notional units and let $S_t\in\R^d$ denote mid-prices in basis-point price units. The dealer receives RFQs indexed by
\[
    i=(m,\tau,s,k),
\]
where $m$ is the bond, $\tau$ is the client or flow segment, $s\in\{b,a\}$ is the side, and $k$ is the size bucket with notional $z_k>0$. A bid RFQ means the client sells and the dealer buys; an ask RFQ means the client buys and the dealer sells.

For RFQ bucket $i$, the baseline arrival intensity is $\lambda_i$ and the quote-dependent fill probability is $f_i(\delta_i)$, where $\delta_i$ is the dealer's quote offset from mid. The executed fill intensity is
\begin{equation}
    \Lambda_i(\delta_i)=\lambda_i f_i(\delta_i).
\end{equation}
Throughout, $f_i$ is decreasing in $\delta_i$; wider quotes win less flow.

The mid-price dynamics are
\begin{equation}
    \dd S_t = \mu_t\,\dd t + \sigma\,\dd W_t,
    \qquad \Sigma=\sigma\sigma^\top.
\end{equation}
For a credit book, a more useful specification is the factor model
\begin{equation}
    \dd S_t = \mu_t\,\dd t + B\,\dd F_t + \dd \epsilon_t,
    \qquad
    \Sigma = B\Omega B^\top+\Sigma_\epsilon,
    \label{eq:factor_cov}
\end{equation}
where $F_t$ contains systematic credit and rates factors and $\epsilon_t$ is idiosyncratic spread/price noise.

\subsection{Cash and inventory jumps}

If a bid RFQ in bond $m$ and size $z_k$ is filled, the dealer buys at $S^m_t-\delta_i$ and inventory changes by $+z_k e_m$. If an ask RFQ is filled, the dealer sells at $S^m_t+\delta_i$ and inventory changes by $-z_k e_m$. Let
\begin{equation}
    \Delta_i Q =
    \begin{cases}
        +z_k e_m, & s=b,\\
        -z_k e_m, & s=a.
    \end{cases}
\end{equation}

We use a reduced value function of the form
\begin{equation}
    V(t,x,q,s)=x+q^\top s+u(t,q).
    \label{eq:value_ansatz}
\end{equation}
This linear cash/mark-to-market reduction is standard in risk-neutral inventory-control formulations with quadratic inventory penalties.

Define the per-notional inventory opportunity cost
\begin{equation}
    p_i[u](t,q)
    =\frac{u(t,q)-u(t,q+\Delta_i Q)}{z_k}.
    \label{eq:p_def}
\end{equation}
For a bid fill, this is $(u(q)-u(q+z_ke_m))/z_k$; for an ask fill, it is $(u(q)-u(q-z_ke_m))/z_k$.

\subsection{Control objective and terminal condition}

For completeness, the finite-horizon objective underlying the reduced HJB is
\begin{align}
\sup_{\delta}\,\E\Bigg[&X_T+Q_T^\top S_T-g(Q_T) \\
&-\int_0^T\left(\frac{\phi}{2}Q_t^\top\Sigma Q_t+
\sum_{\tau\in\mathcal A}\frac{\kappa_\tau W^Q_\tau}{2}
\left(r^Q_\tau(\delta_t)-r^{Q,\star}_\tau\right)^2\right)\,\dd t\Bigg].
\label{eq:full_objective}
\end{align}
Here $g$ is a terminal inventory or liquidation penalty. The reduced terminal condition is therefore
\begin{equation}
    u(T,q)=-g(q).
    \label{eq:terminal_condition}
\end{equation}
The simulations report terminal liquidation costs explicitly; the HJB derivation below is written for a generic terminal penalty $g$.

\section{Residual toxicity and credit alpha}

A raw adverse markout is not the same as client toxicity. Let trade $j$ occur in bucket $i(j)$ and let $M^{raw}_{j,h}$ be the dealer-signed post-trade markout over horizon $h$, positive when adverse to the dealer. We decompose
\begin{equation}
    M^{raw}_{j,h}
    =
    M^{factor}_{j,h}
    +M^{carry/rolldown}_{j,h}
    +M^{issuerRV}_{j,h}
    +M^{index/ETF}_{j,h}
    +\varepsilon_{j,h}.
    \label{eq:markout_decomp}
\end{equation}
Only the residual component $\varepsilon_{j,h}$ should be used to define client-flow toxicity. The residual-toxicity cost for bucket $i$ is
\begin{equation}
    \chi_i^{res}
    =
    \E\left[M^{res}_{j,h}\mid i(j)=i,\,\text{fill}\right],
    \qquad
    M^{res}_{j,h}=\varepsilon_{j,h}.
    \label{eq:resid_tox}
\end{equation}
A positive $\chi_i^{res}$ means the fill has expected residual adverse-selection cost after controlling for observable credit and index effects.

Credit alpha enters separately through $\mu_t$. In a credit setting, one may write
\begin{equation}
    \mu^m_t
    =
    \mu^{carry}_{m,t}
    +\mu^{rolldown}_{m,t}
    +\mu^{issuerRV}_{m,t}
    +\mu^{flow}_{m,t}
    +\mu^{index}_{m,t}.
    \label{eq:credit_alpha}
\end{equation}
This separation is important: a fund that buys bonds with strong carry/rolldown or positive issuer-RV convergence may generate adverse raw markouts for a dealer who underprices those signals, but that is not residual toxicity.

\section{Quality-adjusted hit ratio}

For each targeted client tier or flow segment $\tau$, define a quality weight
\begin{equation}
    w_i=\exp(-\gamma \chi_i^{res}),
    \qquad i\in\mathcal I_\tau,
    \label{eq:weight}
\end{equation}
where $\gamma\geq0$ controls the strength of the quality adjustment. Other monotone bounded weights can be used; the exponential form is convenient and positive.

The quality-adjusted instantaneous hit ratio is
\begin{equation}
    r^Q_\tau(\delta)
    =
    \frac{\sum_{i\in\mathcal I_\tau} z_i w_i\Lambda_i(\delta_i)}
    {W^Q_\tau},
    \qquad
    W^Q_\tau=
    \sum_{i\in\mathcal I_\tau} z_i w_i\lambda_i.
    \label{eq:quality_hr}
\end{equation}
This normalization keeps $r^Q_\tau$ on the same scale as an ordinary hit ratio.

\begin{assumption}[Fill curves]
For every RFQ bucket $i$, $f_i$ is continuously differentiable, strictly decreasing, and satisfies $0<f_i(\delta)<1$. The Hamiltonian $H_i(p)=\sup_\delta \Lambda_i(\delta)(\delta-p)$ is finite and twice differentiable on the relevant range of $p$.
\end{assumption}

\begin{assumption}[Residual quality weights]
The weights $w_i$ are non-negative and bounded. The exponential specification \eqref{eq:weight} is used in the simulations, but the separability result only requires that $w_i$ be fixed at the quote-decision time and independent of the quote offset $\delta_i$.
\end{assumption}

\section{Reduced HJB and exact dualization}

The dealer maximizes expected mark-to-market wealth, net of running inventory risk and hit-ratio target penalties. With residual toxicity charged per unit notional on filled RFQs, the reduced HJB is
\begin{align}
0=&\,\partial_t u(t,q)+q^\top\mu_t-\frac{\phi}{2}q^\top\Sigma q \notag\\
&+\sup_{\delta}
\Bigg\{
\sum_i z_i\Lambda_i(\delta_i)
\left(\delta_i-p_i[u](t,q)-\chi_i^{res}\right)
-
\sum_{\tau\in\mathcal A}
\frac{\kappa_\tau W^Q_\tau}{2}
\left(r^Q_\tau(\delta)-r^{Q,\star}_\tau\right)^2
\Bigg\}.
\label{eq:hjb_raw}
\end{align}
Here $\mathcal A$ is the set of targeted tiers and $\kappa_\tau>0$ is the target penalty strength.

The negative quadratic penalty admits the scalar dual representation
\begin{equation}
    -\frac{\kappa W}{2}(r-r^\star)^2
    =
    \inf_{\xi\in\R}
    W\left[\xi(r-r^\star)+\frac{\xi^2}{2\kappa}\right].
    \label{eq:dual_identity}
\end{equation}
The pointwise identity rewrites the statewise penalty as a saddle expression. In the dual formulation used below, quote controls are maximized for a fixed scalar multiplier and the resulting one-dimensional dual objective is minimized over that multiplier. Under the fill-curve regularity assumptions above, this scalar objective is strictly convex, and its minimizer satisfies the saddle first-order condition of the penalized statewise problem. For fixed dual variables $\xi_\tau$, the HJB separates across RFQ buckets.

\begin{proposition}[Separable exact-dual Hamiltonian]
Let
\begin{equation}
    H_i(p)=\sup_{\delta\in\R}\Lambda_i(\delta)(\delta-p).
    \label{eq:Hamiltonian}
\end{equation}
After dualizing each quality-hit-ratio penalty, the RFQ contribution for bucket $i\in\mathcal I_\tau$ is
\begin{equation}
    z_i H_i\left(p_i[u](t,q)+\chi_i^{res}-w_i\xi_\tau\right).
    \label{eq:shifted_hamiltonian}
\end{equation}
Thus residual toxicity and the quality-hit-ratio target enter through a scalar shift of the Hamiltonian argument.
\end{proposition}

\begin{proof}
For fixed $\xi_\tau$, the dual term contributes
\[
    W^Q_\tau\xi_\tau r^Q_\tau(\delta)
    =
    \sum_{i\in\mathcal I_\tau} z_i w_i\xi_\tau\Lambda_i(\delta_i).
\]
Combining this with the per-fill reward gives
\[
    z_i\Lambda_i(\delta_i)
    \left(\delta_i-p_i[u]-\chi_i^{res}+w_i\xi_\tau\right)
    =
    z_i\Lambda_i(\delta_i)
    \left(\delta_i-\big(p_i[u]+\chi_i^{res}-w_i\xi_\tau\big)\right).
\]
Taking the supremum over each $\delta_i$ gives \eqref{eq:shifted_hamiltonian}.
\end{proof}

For a targeted tier $\tau$, write $p_i^0=p_i[u]+\chi_i^{res}$. The statewise dual Hamiltonian for that tier is
\begin{equation}
\mathcal H_\tau^Q(u,q)
=
\inf_{\xi_\tau\in\R}\mathcal J_\tau(\xi_\tau),
\label{eq:dual_hamiltonian}
\end{equation}
where
\begin{equation}
\mathcal J_\tau(\xi_\tau)
=
\sum_{i\in\mathcal I_\tau}z_iH_i(p_i^0-w_i\xi_\tau)
+W^Q_\tau\left[-\xi_\tau r^{Q,\star}_\tau+\frac{\xi_\tau^2}{2\kappa_\tau}\right].
\label{eq:dual_objective}
\end{equation}
The term $-W^Q_\tau\xi_\tau r^{Q,\star}_\tau+W^Q_\tau\xi_\tau^2/(2\kappa_\tau)$ is the non-RFQ part of the dual penalty, while the quote-dependent part is embedded in the shifted Hamiltonians. Differentiating \eqref{eq:dual_objective} gives the scalar first-order condition below.

The exact dual variable solves a one-dimensional nonlinear equation. Since $H_i'(p)=-\Lambda_i(\delta_i^\star(p))$, the exact fixed point is
\begin{equation}
    \boxed{
    \xi_\tau
    =
    \kappa_\tau
    \left[
    r^{Q,\star}_\tau
    +
    \frac{1}{W^Q_\tau}
    \sum_{i\in\mathcal I_\tau}
    z_i w_i
    H_i'\left(p_i^0-w_i\xi_\tau\right)
    \right],
    }
    \label{eq:exact_xi}
\end{equation}
where $p_i^0=p_i[u]+\chi_i^{res}$. Equivalently,
\begin{equation}
    \xi_\tau=\kappa_\tau\left(r^{Q,\star}_\tau-r^Q_\tau(\xi_\tau)\right).
    \label{eq:xi_target_gap}
\end{equation}
Because $r^Q_\tau(\xi_\tau)$ is increasing in $\xi_\tau$ under decreasing fill curves, the equation has a unique solution under mild regularity. Numerically it is solved by a one-dimensional nonlinear root solver. No linearization is required for the main algorithm.

\begin{proposition}[Existence and uniqueness of the exact dual variable]
Fix $(t,q)$ and a targeted tier $\tau$. Suppose $H_i$ is twice differentiable and convex, with $H_i''\geq 0$, and suppose at least one bucket in $\mathcal I_\tau$ has $w_i>0$. Define
\begin{equation}
    F_\tau(\xi)=
    \xi-
    \kappa_\tau
    \left[
    r^{Q,\star}_\tau+
    \frac{1}{W^Q_\tau}
    \sum_{i\in\mathcal I_\tau}z_iw_iH_i'(p_i^0-w_i\xi)
    \right].
\end{equation}
Then
\begin{equation}
    F_\tau'(\xi)=
    1+
    \frac{\kappa_\tau}{W^Q_\tau}
    \sum_{i\in\mathcal I_\tau}z_iw_i^2H_i''(p_i^0-w_i\xi)>0.
\end{equation}
Hence $F_\tau$ is strictly increasing and the exact dual equation has at most one root. Under the usual logistic or exponential fill specifications and any target $r^{Q,\star}_\tau\in[0,1]$, $F_\tau(\xi)\to-\infty$ as $\xi\to-\infty$ and $F_\tau(\xi)\to+\infty$ as $\xi\to+\infty$. Therefore a unique root exists and can be found robustly by a scalar root solver.
\end{proposition}

\section{Closed-form quote map for logistic fill curves}

Assume
\begin{equation}
    f_i(\delta)=\frac{1}{1+\exp(\alpha_i+\beta_i\delta)},
    \qquad \beta_i>0.
    \label{eq:logistic_fill}
\end{equation}
Let $x_i(p)$ be
\begin{equation}
    x_i(p)=W_0\left(\exp(-\alpha_i-\beta_i p-1)\right),
    \label{eq:lambert_x}
\end{equation}
where $W_0$ is the principal branch of Lambert's $W$ function. The optimal quote is
\begin{equation}
    \boxed{
    \delta_i^\star(p)=p+\frac{1+x_i(p)}{\beta_i}.
    }
    \label{eq:delta_star}
\end{equation}
Moreover,
\begin{equation}
    H_i'(p)=-\lambda_i\frac{x_i(p)}{1+x_i(p)},
    \qquad
    H_i''(p)=\lambda_i\beta_i\frac{x_i(p)}{(1+x_i(p))^3}.
    \label{eq:H_derivatives}
\end{equation}

\section{Quadratic approximation and quote decomposition}

For multi-bond portfolios, the exact inventory-grid HJB is not scalable. We use the quadratic approximation
\begin{equation}
    u(t,q)\approx -\frac12 q^\top A(t)q+\ell(t)^\top q-C(t),
    \label{eq:quadratic_u}
\end{equation}
with symmetric positive semidefinite $A(t)$.

For a bid in bond $m$ and size $z_k$,
\begin{equation}
    p^{b,k}_m(t,q)
    =e_m^\top A(t)\left(q+\frac{z_k}{2}e_m\right)-\ell_m(t).
    \label{eq:p_bid_quad}
\end{equation}
For an ask,
\begin{equation}
    p^{a,k}_m(t,q)
    =e_m^\top A(t)\left(-q+\frac{z_k}{2}e_m\right)+\ell_m(t).
    \label{eq:p_ask_quad}
\end{equation}
The full quote argument is therefore
\begin{equation}
    \bar p_i(t,q)
    =p_i(t,q)+\chi_i^{res}-w_i\xi_\tau.
    \label{eq:full_argument}
\end{equation}
The dealer quotes $\delta_i^\star(\bar p_i)$ using \eqref{eq:delta_star} or the corresponding Hamiltonian maximizer.

Under the symmetric local quadratic expansion used in the BEGV/GLFT tradition, the inventory curvature satisfies
\begin{equation}
    A'(t)=A(t)D A(t)-\phi\Sigma,
    \label{eq:riccati}
\end{equation}
where $D$ is the diagonal opportunity matrix with entries of the form
\begin{equation}
    D_{mm}=\sum_{i:m(i)=m} z_i H_i''(\bar p_i^{ref}).
    \label{eq:D_def}
\end{equation}
The reference point $\bar p_i^{ref}$ may include residual toxicity and the current target-dual operating point. In stationary form, $A D A=\phi\Sigma$ and
\begin{equation}
    \boxed{
    A=
    \sqrt{\phi}\,
    D^{-1/2}
    \left(D^{1/2}\Sigma D^{1/2}\right)^{1/2}
    D^{-1/2}.
    }
    \label{eq:A_stationary}
\end{equation}
The linear credit-alpha term is
\begin{equation}
    \boxed{\ell=(AD)^{-1}\mu.}
    \label{eq:ell_stationary}
\end{equation}

Combining these pieces gives the quote decomposition
\begin{equation}
\boxed{
\begin{aligned}
\text{quote offset}
&=\text{riskless spread}+\text{inventory skew}-\text{credit-alpha skew}\\
&\quad +\text{residual-toxicity charge}-\text{quality-hit-ratio subsidy}.
\end{aligned}}
\label{eq:quote_decomposition}
\end{equation}
This is the main practitioner-facing formula.

\begin{remark}[Sign conventions]
A positive residual-toxicity cost $\chi_i^{res}$ widens the quote by increasing the Hamiltonian argument. A positive quality-dual variable $\xi_\tau$ tightens quotes for buckets with large $w_i$, because the argument is shifted by $-w_i\xi_\tau$. In a one-bond or diagonal-risk setting, positive credit alpha produces $\ell_m>0$ in the stationary approximation; this lowers the bid argument and raises the ask argument, so the dealer bids more aggressively and is less eager to sell a positive-alpha bond. In a correlated multi-bond book, $\ell=(AD)^{-1}\mu$ is a risk-adjusted alpha tilt and component signs can also reflect factor hedging and cross-bond correlations.
\end{remark}

\section{Numerical algorithm}

The numerical algorithm is:
\begin{enumerate}[leftmargin=*]
    \item Estimate or specify $\lambda_i$, $f_i$, $\chi_i^{res}$, $w_i$, $\mu$, and $\Sigma$.
    \item Compute the opportunity matrix $D$ and stationary $A$ using \eqref{eq:A_stationary}; compute $\ell$ from \eqref{eq:ell_stationary}.
    \item At each quote state $q$, compute $p_i^0=p_i(q)+\chi_i^{res}$ for all relevant RFQ buckets.
    \item For each targeted tier $\tau$, solve the exact scalar nonlinear equation \eqref{eq:exact_xi}.
    \item Quote $\delta_i^\star(p_i^0-w_i\xi_\tau)$.
    \item In simulation, draw RFQ arrivals, fills, mid-price moves, and update cash and inventory.
\end{enumerate}

The only nonlinear solve in the multi-bond simulation is the scalar dual equation \eqref{eq:exact_xi}. This is a significant simplification relative to a full multi-dimensional inventory-grid HJB.

\section{Synthetic numerical experiments}

\subsection{Design}

The numerical experiments are synthetic. They are intended to test mechanisms, not to claim empirical performance. A synthetic EUR IG universe is generated with $d=20$ bonds, sector/rating/maturity attributes, and a factor covariance model of the form \eqref{eq:factor_cov}. The core flow universe contains two main client-flow types:
\begin{enumerate}[leftmargin=*]
    \item \textbf{Carry/RV flow}: low residual toxicity after controlling for carry, rolldown, and issuer-RV effects.
    \item \textbf{Residual-toxic flow}: high residual adverse markout after the same controls.
\end{enumerate}
The policies compared are raw hit-ratio targeting, residual-quality-adjusted targeting, and several benchmark policies. The headline comparison is between raw hit-ratio targeting and residual-quality-adjusted targeting with credit-alpha drift.

\paragraph{Nonlinear dual implementation.}
The simulations in this section use the nonlinear scalar dual equation \eqref{eq:exact_xi}. To keep the Monte Carlo tractable, the Hamiltonian quote and fill maps are precomputed on a dense grid and the scalar nonlinear dual equation is solved by Newton iterations initialized by the local solution. This is not the linearized dual closure used for preliminary sweeps. The Monte Carlo reported here uses 300 paths per target/policy for the main frontier and 100 paths per policy for the 8\% benchmark run over five trading days. Larger runs can be obtained by increasing the path count in the same simulation design.

\subsection{Attained service/economics frontier}

The target parameter $r^\star$ is a penalty target rather than a hard equality constraint. Therefore, the main frontier is plotted against the \emph{attained} residual-quality hit ratio, not merely the target label. Figure~\ref{fig:attained_frontier} reports mean net PnL per day against attained quality hit ratio, with 95\% confidence intervals for PnL. Table~\ref{tab:attained_frontier} reports the same frontier points.

\begin{table}[H]
\centering
\caption{Nonlinear-dual attained service/economics frontier. PnL is bp$\cdot$MM per day.}
\label{tab:attained_frontier}
\small
\begin{tabular}{lrrrr}
\toprule
Policy & Target & Attained quality HR & PnL/day & 95\% CI \\
\midrule
Raw target & 4.00\% & 4.14\% & 7.32 & [6.08, 8.57] \\
Raw target & 6.00\% & 5.64\% & -2.66 & [-4.11, -1.21] \\
Raw target & 8.00\% & 7.21\% & -15.88 & [-17.46, -14.31] \\
Raw target & 10.00\% & 8.87\% & -33.33 & [-34.88, -31.78] \\
Raw target & 12.00\% & 10.65\% & -55.36 & [-57.08, -53.64] \\
Residual QA + alpha & 4.00\% & 3.56\% & 14.98 & [13.89, 16.08] \\
Residual QA + alpha & 6.00\% & 4.89\% & 12.11 & [10.93, 13.28] \\
Residual QA + alpha & 8.00\% & 6.38\% & 5.74 & [4.42, 7.07] \\
Residual QA + alpha & 10.00\% & 7.96\% & -2.58 & [-3.85, -1.31] \\
Residual QA + alpha & 12.00\% & 9.64\% & -14.22 & [-15.55, -12.88] \\
\bottomrule
\end{tabular}
\end{table}

\begin{figure}[H]
\centering
\includegraphics[width=0.78\textwidth]{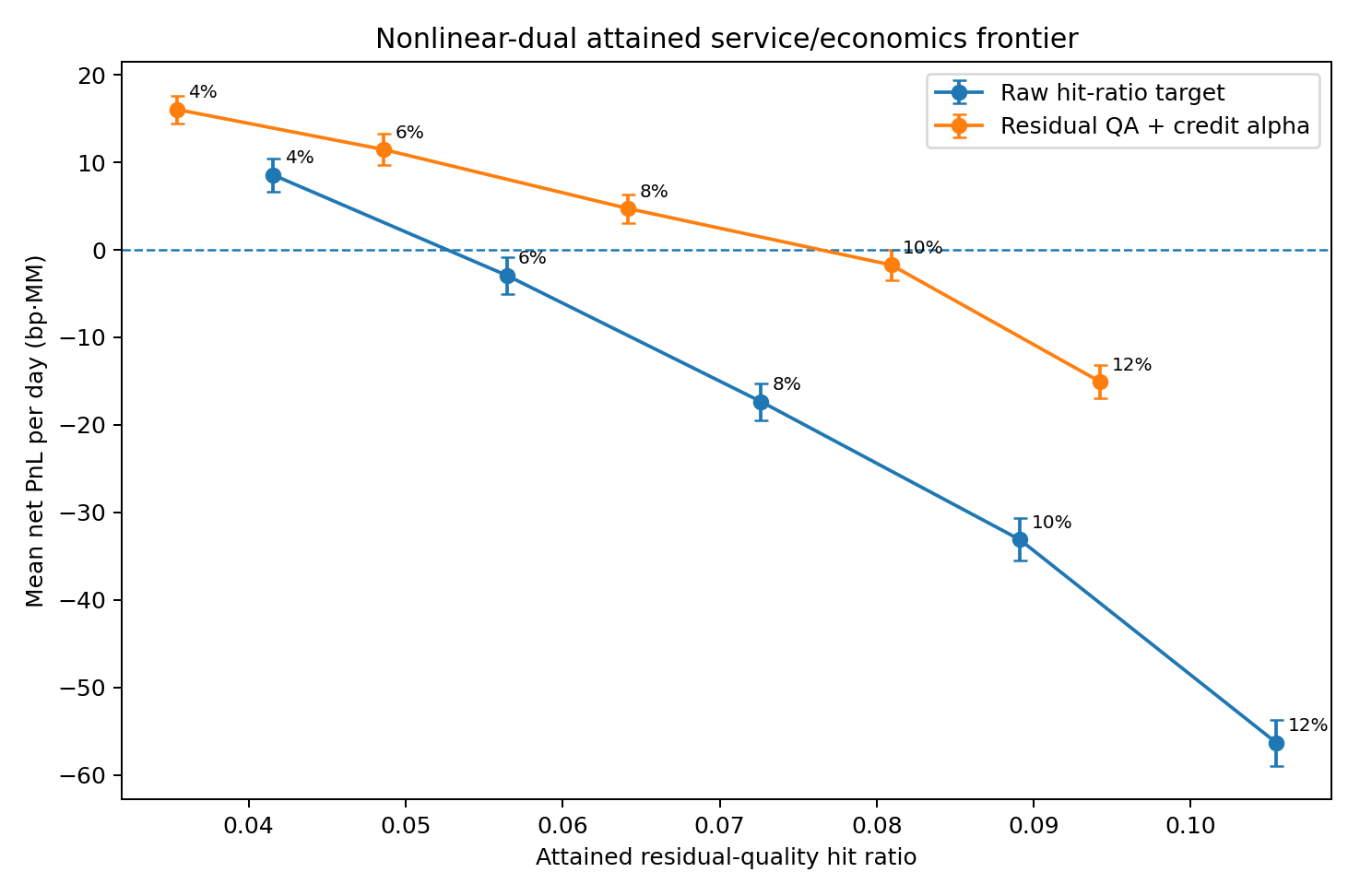}
\caption{Nonlinear-dual attained service/economics frontier. The x-axis is realized residual-quality hit ratio, not the target parameter.}
\label{fig:attained_frontier}
\end{figure}

To compare the two policies at matched attained service levels, Table~\ref{tab:matched_quality} linearly interpolates the two frontiers on their overlapping quality-hit-ratio range. Residual-quality targeting with credit alpha dominates raw targeting throughout the overlap in this synthetic run.

\begin{table}[H]
\centering
\caption{Matched attained-quality comparison. PnL is bp$\cdot$MM per day.}
\label{tab:matched_quality}
\small
\begin{tabular}{rrrr}
\toprule
Quality HR & Raw PnL/day & Residual QA + alpha PnL/day & Saving/day \\
\midrule
4.14\% & 7.32 & 13.73 & 6.41 \\
4.83\% & 2.75 & 12.24 & 9.49 \\
5.51\% & -1.83 & 9.44 & 11.27 \\
6.20\% & -7.38 & 6.52 & 13.90 \\
6.89\% & -13.15 & 3.07 & 16.23 \\
7.58\% & -19.70 & -0.55 & 19.15 \\
8.26\% & -26.95 & -4.68 & 22.27 \\
8.95\% & -34.35 & -9.45 & 24.90 \\
9.64\% & -42.85 & -14.22 & 28.63 \\
\bottomrule
\end{tabular}
\end{table}

\begin{figure}[H]
\centering
\includegraphics[width=0.72\textwidth]{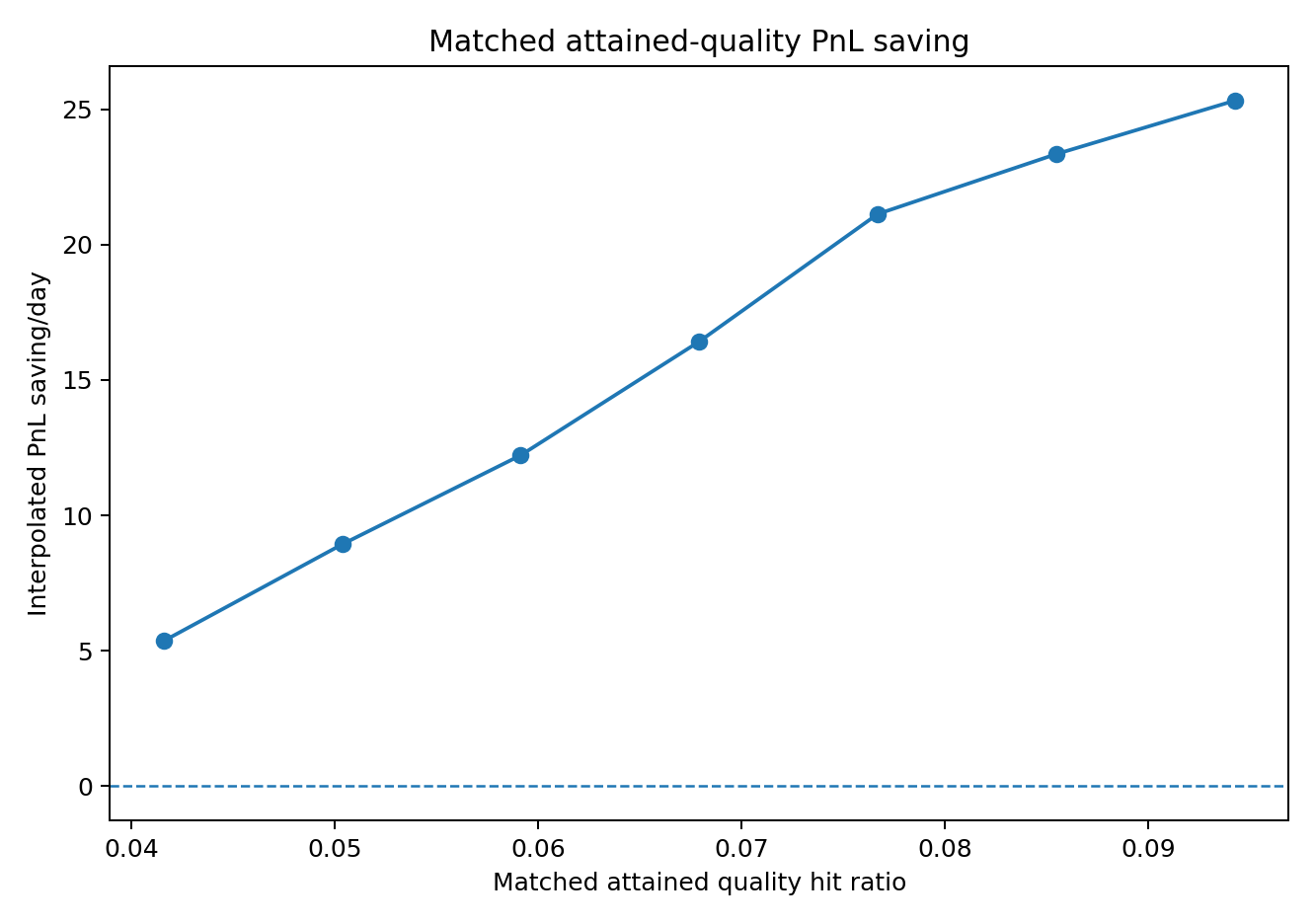}
\caption{Interpolated PnL saving of residual-quality targeting over raw targeting at matched attained quality-hit-ratio levels.}
\label{fig:matched_quality_saving}
\end{figure}

\subsection{Cost of raw hit ratio: PnL attribution}

Table~\ref{tab:pnl_attribution} decomposes the PnL difference between residual-quality targeting with credit alpha and raw hit-ratio targeting at the 8\% target. The key mechanism is that residual-quality targeting gives up some spread capture but saves much more in residual adverse-selection cost.

\begin{table}[H]
\centering
\caption{PnL attribution at the 8\% target: Residual QA + alpha minus Raw target. PnL is bp$\cdot$MM per day. For cost rows, the contribution is reported as cost saved.}
\label{tab:pnl_attribution}
\small
\begin{tabular}{lrrr}
\toprule
Component & Raw & Residual QA + alpha & Contribution \\
\midrule
Spread capture & 34.03 & 24.54 & -9.49 \\
Credit-alpha PnL & -2.30 & -1.17 & 1.14 \\
Factor PnL & 1.10 & -0.80 & -1.90 \\
Idiosyncratic PnL & 0.60 & 0.49 & -0.11 \\
Residual toxicity cost & 41.73 & 13.46 & 28.26 \\
Terminal liquidation cost & 6.02 & 4.35 & 1.67 \\
Net PnL saving & -14.32 & 5.25 & 19.57 \\
\bottomrule
\end{tabular}
\end{table}

\begin{figure}[H]
\centering
\includegraphics[width=0.78\textwidth]{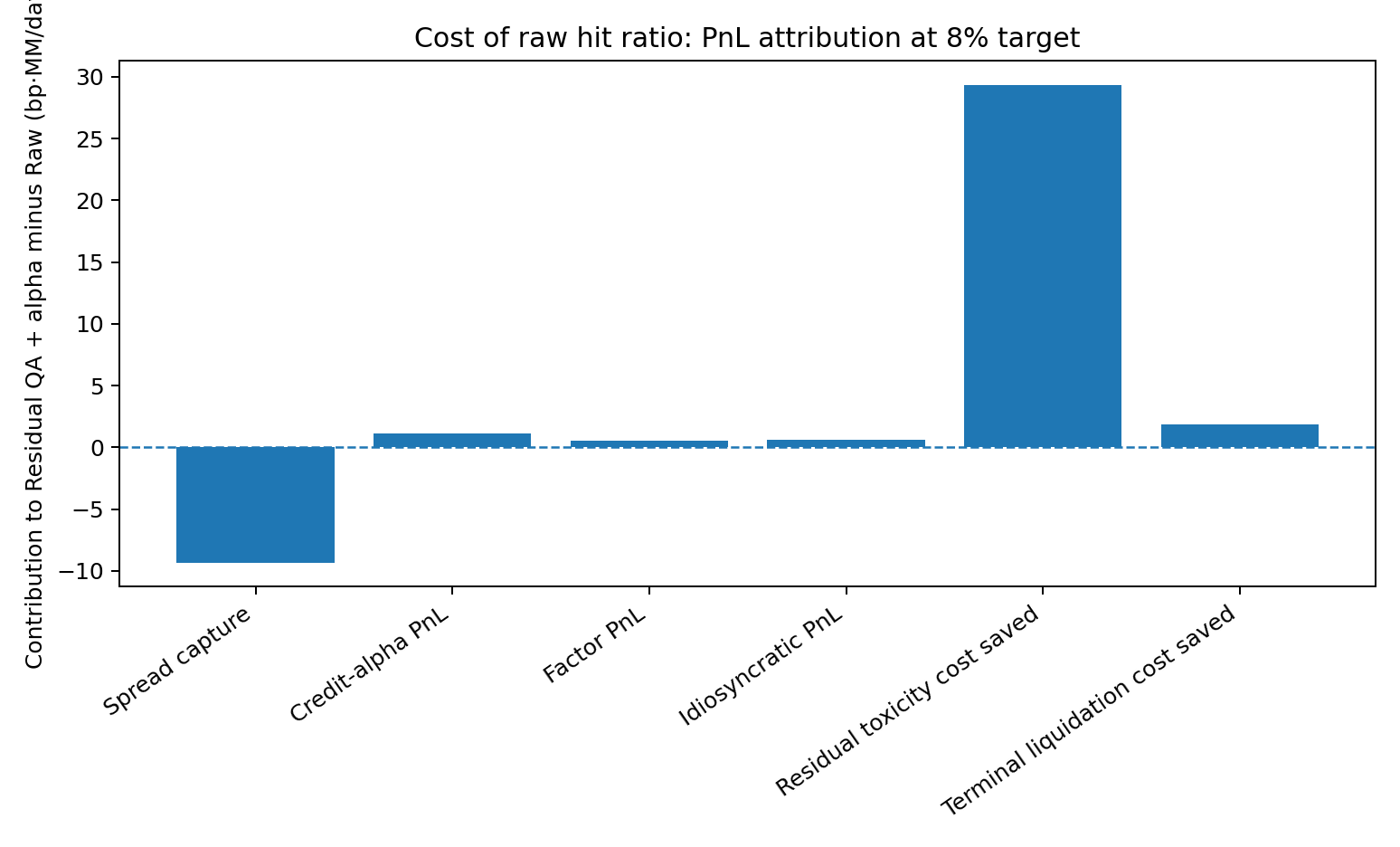}
\caption{Cost of raw hit ratio at the 8\% target. Residual-quality targeting gives up some spread capture but saves substantially more residual-toxicity cost.}
\label{fig:pnl_attribution}
\end{figure}

\subsection{Flow allocation and residual-cost diagnostics}

Table~\ref{tab:flow_allocation} reports how policies allocate service between carry/RV flow and residual-toxic flow at the 8\% target. The raw target wins a similar or higher hit ratio on residual-toxic flow than on carry/RV flow. Residual-quality targeting keeps service to low-residual-toxicity carry/RV flow while sharply reducing residual-toxic wins.

\begin{table}[H]
\centering
\caption{Flow allocation at the 8\% target. Residual cost is bp$\cdot$MM per day.}
\label{tab:flow_allocation}
\small
\begin{tabular}{lrrrr}
\toprule
Policy & Quality HR & Carry/RV HR & Residual-toxic HR & Residual cost/day \\
\midrule
Raw target & 7.37\% & 7.29\% & 7.47\% & 41.73 \\
Residual QA + alpha & 6.57\% & 7.05\% & 0.87\% & 13.46 \\
Residual-toxicity QA & 6.63\% & 7.12\% & 0.83\% & 13.37 \\
Risk only & 4.78\% & 4.75\% & 4.80\% & 27.13 \\
Naive gross toxicity & 1.06\% & 1.08\% & 0.61\% & 4.00 \\
\bottomrule
\end{tabular}
\end{table}

\begin{figure}[H]
\centering
\includegraphics[width=0.78\textwidth]{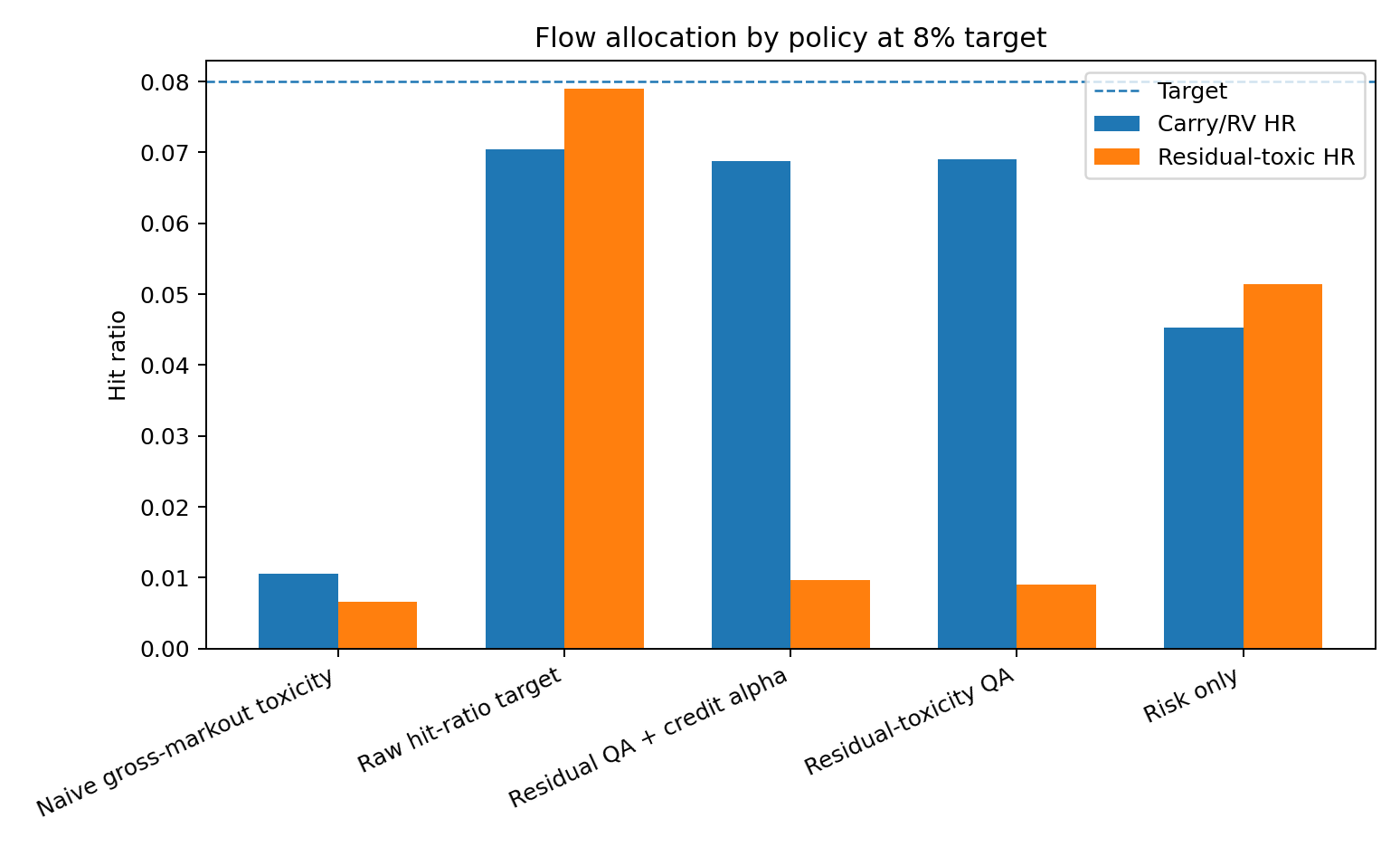}
\caption{Flow allocation by policy at the 8\% target. Residual-quality targeting reallocates wins away from residual-toxic flow.}
\label{fig:flow_allocation}
\end{figure}

\subsection{Distributional and inventory-risk analysis}

Mean PnL alone is not sufficient for a market-making problem. Table~\ref{tab:pnl_distribution} reports the PnL distribution at the 8\% target, and Table~\ref{tab:inventory_risk} reports terminal inventory-risk measures. The residual-quality policy improves PnL relative to raw targeting without increasing terminal inventory risk.

\begin{table}[H]
\centering
\caption{PnL distribution at the 8\% target. Units are bp$\cdot$MM per day.}
\label{tab:pnl_distribution}
\small
\begin{tabular}{lrrrrr}
\toprule
Policy & Mean & Std & 5\% & Median & 95\% \\
\midrule
Raw target & -14.32 & 13.21 & -35.14 & -16.85 & 7.92 \\
Residual QA + alpha & 5.25 & 10.78 & -12.50 & 5.17 & 23.52 \\
Residual-toxicity QA & 4.54 & 11.12 & -13.47 & 3.48 & 26.03 \\
Risk only & 5.04 & 12.38 & -13.81 & 3.32 & 28.59 \\
Naive gross toxicity & 8.22 & 8.09 & 1.13 & 6.44 & 21.31 \\
\bottomrule
\end{tabular}
\end{table}

\begin{table}[H]
\centering
\caption{Inventory-risk measures at the 8\% target.}
\label{tab:inventory_risk}
\small
\begin{tabular}{lrrrr}
\toprule
Policy & Mean PnL/day & PnL std/day & Terminal inv. risk & PnL / inv. risk \\
\midrule
Raw target & -14.32 & 13.21 & 25.73 & -0.56 \\
Residual QA + alpha & 5.25 & 10.78 & 20.08 & 0.26 \\
Residual-toxicity QA & 4.54 & 11.12 & 22.20 & 0.20 \\
Risk only & 5.04 & 12.38 & 22.72 & 0.22 \\
Naive gross toxicity & 8.22 & 8.09 & 9.02 & 0.91 \\
\bottomrule
\end{tabular}
\end{table}

\begin{figure}[H]
\centering
\includegraphics[width=0.72\textwidth]{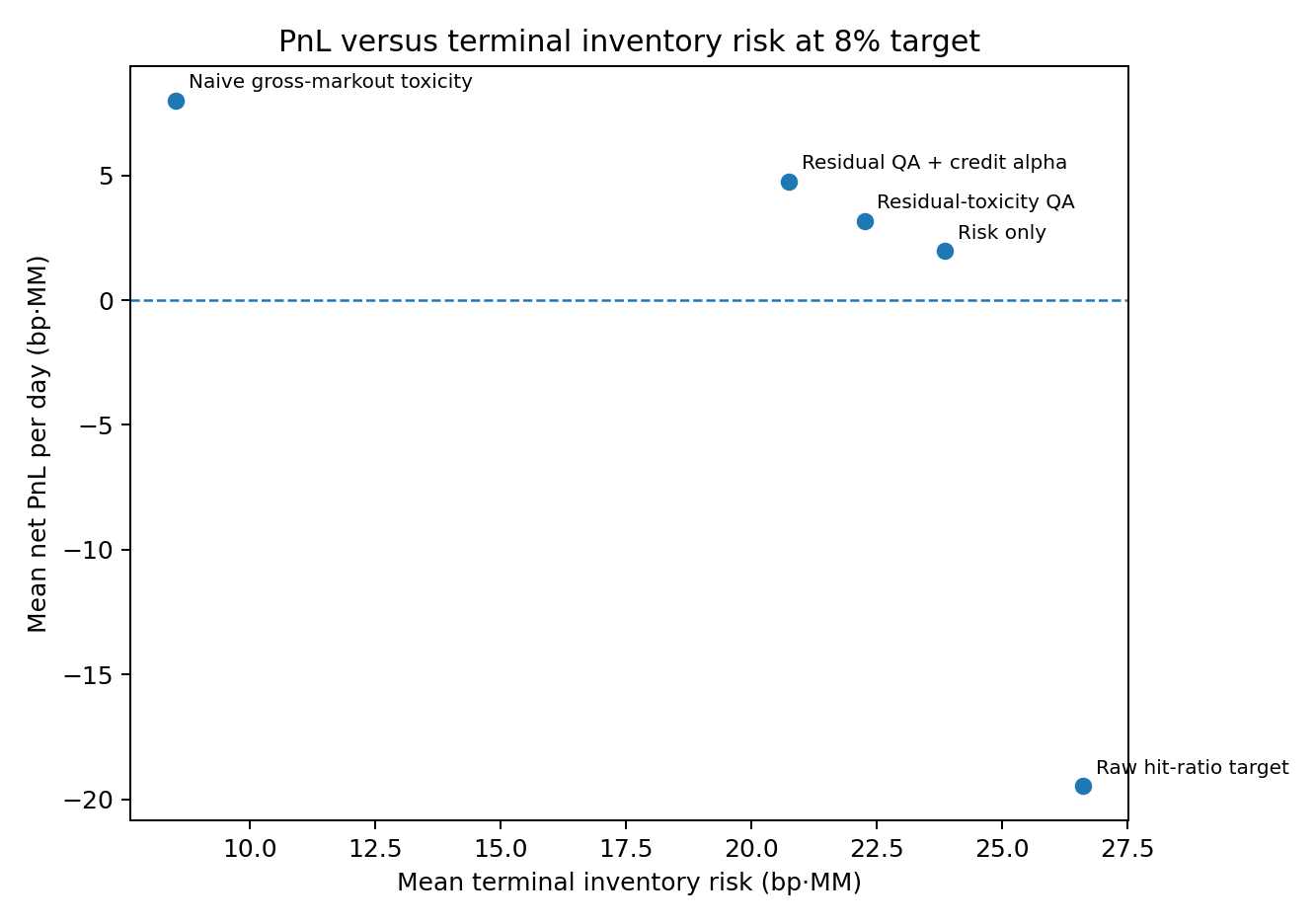}
\caption{Mean PnL versus terminal inventory risk at the 8\% target.}
\label{fig:inventory_risk}
\end{figure}

\subsection{Liquidation-cost sensitivity}

Terminal liquidation cost is a modelling assumption. To check that the PnL saving is not an artefact of a single liquidation parameter, Table~\ref{tab:liq_sensitivity} reports the PnL saving of residual-quality targeting over raw targeting at the 8\% target across several terminal liquidation cost levels. Figure~\ref{fig:liq_sensitivity} reports the same sensitivity across all target levels.

\begin{table}[H]
\centering
\caption{Liquidation-cost sensitivity at the 8\% target. Saving is Residual QA + alpha minus Raw target, in bp$\cdot$MM per day.}
\label{tab:liq_sensitivity}
\small
\begin{tabular}{rrr}
\toprule
Liquidation base bp & Saving/day & 95\% CI \\
\midrule
0.00 & 19.77 & [18.54, 21.00] \\
0.25 & 21.10 & [19.88, 22.32] \\
0.50 & 22.42 & [21.21, 23.64] \\
1.00 & 25.08 & [23.85, 26.30] \\
2.00 & 30.38 & [29.06, 31.70] \\
\bottomrule
\end{tabular}
\end{table}

\begin{figure}[H]
\centering
\includegraphics[width=0.78\textwidth]{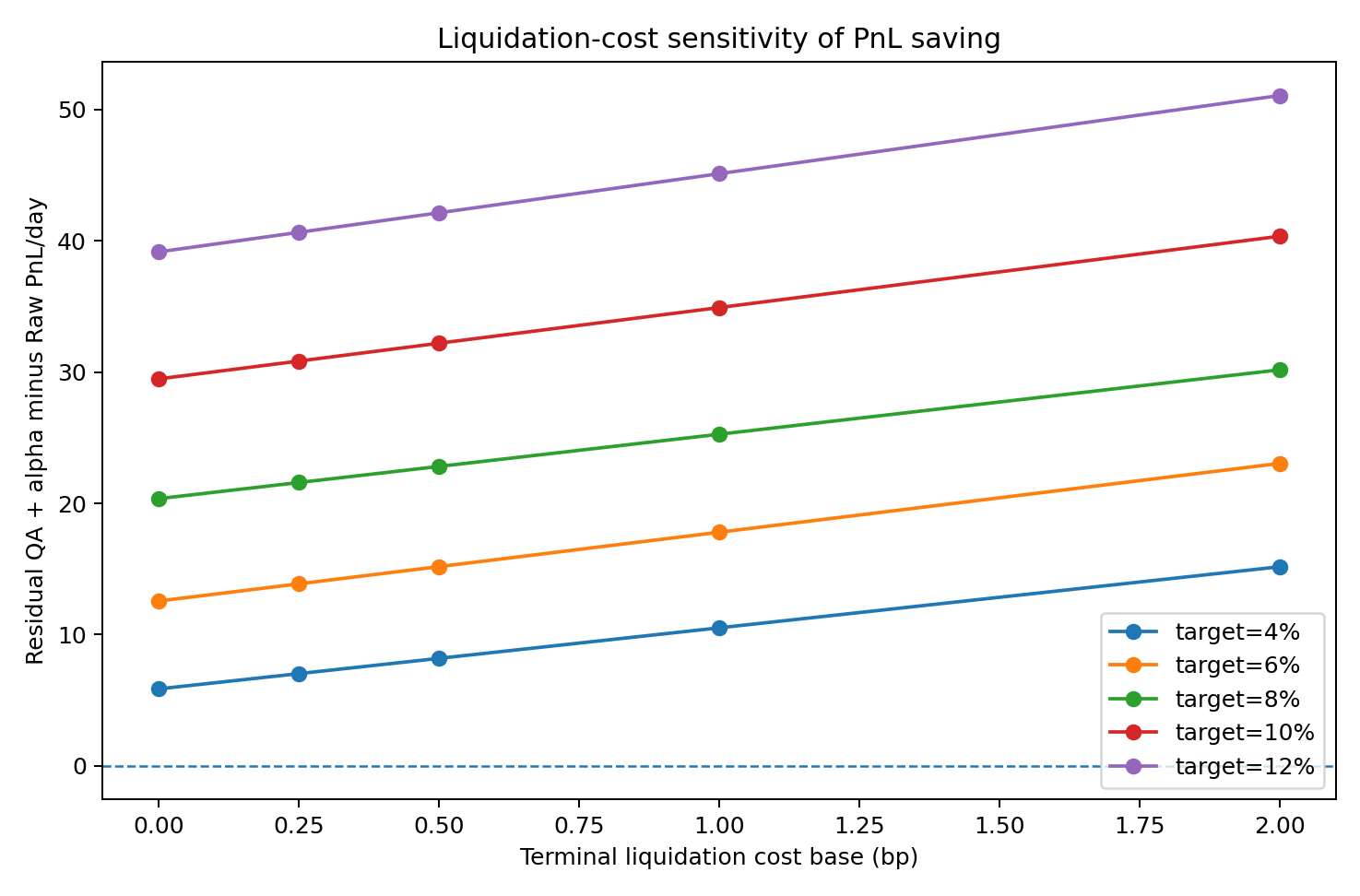}
\caption{Liquidation-cost sensitivity of the residual-quality PnL saving across target levels.}
\label{fig:liq_sensitivity}
\end{figure}
\section{Extension: style-aware client-flow warehousing}
\label{sec:style_warehousing}

Forecastable passive/index demand is one example of a broader mechanism: a client franchise can create value when it gives the dealer access to inventory aligned with the dealer's own credit-style alpha. The passive/index experiment is reported in Appendix~\ref{app:passive_details} as a special case. This section focuses on the more general mechanism, style-aware client-flow warehousing. It remains separate from the main RFQ hit-rate HJB: Sweep or portfolio-trade participation is a second execution problem with random fills. The purpose is to show how the same quadratic value approximation can size participation in a way that is consistent with inventory risk and alpha alignment.

Let $\alpha_m$ denote the dealer's credit-style alpha for bond $m$, combining carry, rolldown, issuer-relative-value, liquidity, and flow signals. Positive $\alpha_m$ means that the dealer prefers to be long bond $m$; negative $\alpha_m$ means that the dealer prefers to be short or underweight. Let $Y_m$ denote a signed client-inventory opportunity, with
\begin{equation}
    Y_m>0
    \quad\Longleftrightarrow\quad
    \text{the dealer buys bond }m\text{ from the client}.
\end{equation}
Flow is style-aligned when
\begin{equation}
    Y_m\alpha_m>0,
\end{equation}
meaning that the dealer buys a bond it wants to own, or sells a bond it prefers to be short. To avoid look-ahead bias, the dealer does not observe the true data-generating alpha. It trades on a forecast
\begin{equation}
    \widehat\alpha_m
    =
    \rho_\alpha \alpha_m
    +
    \sqrt{1-\rho_\alpha^2}\,\varepsilon_m,
\end{equation}
where $\rho_\alpha$ measures style-alpha forecast skill.

Conditional on a participation decision $\pi_m\in[0,1]$, the realised fill is random:
\begin{equation}
    \Delta q_m^{\mathrm{sweep}}
    =
    \varphi_m\pi_mY_m,
    \qquad 0\leq \varphi_m\leq 1.
    \label{eq:risk_sweep_random_fill}
\end{equation}
The random fraction $\varphi_m$ captures line-allocation uncertainty, partial fills, and the fact that the dealer does not necessarily receive the exact desired inventory.

Using the quadratic approximation
\begin{equation}
    u(q)\approx -\frac12 q^\top A q+\ell^\top q,
\end{equation}
the expected incremental value of a one-bond Sweep fill is approximated by
\begin{equation}
\begin{split}
    \E[\Delta V_m]
    \approx&\; \bar\varphi\,\pi_m
    \left[
        Y_m(\ell_m-(Aq)_m)+|Y_m|(b_m-c_m)
    \right] \\
    &-\frac12\overline{\varphi^2}\,\pi_m^2Y_m^2A_{mm},
\end{split}
\end{equation}
where $b_m$ is the per-notional Sweep spread benefit, $c_m$ collects acquisition and residual costs, $\bar\varphi=\E[\varphi_m]$, and $\overline{\varphi^2}=\E[\varphi_m^2]$. This yields the diagonal risk-aware participation rule
\begin{equation}
    \pi_m^\star
    =
    \left[
    \frac{
    \bar\varphi\left(Y_m(\ell_m-(Aq)_m)+|Y_m|(b_m-c_m)\right)
    }{
    \overline{\varphi^2}Y_m^2A_{mm}
    }
    \right]_0^{\pi_{\max}},
    \label{eq:risk_aware_sweep_rule}
\end{equation}
where $[\cdot]_0^{\pi_{\max}}$ denotes clipping. The rule increases participation when offered inventory is aligned with risk-adjusted alpha and current inventory, but reduces participation when incremental inventory risk or execution uncertainty is high.

A full treatment would add a second jump-control term,
\begin{equation}
    \lambda^{\mathrm{sweep}}
    \E_{Y,\varphi}
    \left[
        \sup_{\pi\in[0,1]^d}
        \left(
            C(\pi,Y,\varphi)
            +u(t,q+\varphi\odot\pi\odot Y)-u(t,q)
        \right)
    \right],
\end{equation}
to the HJB. The present experiment does not solve this full problem. It uses the local rule in \eqref{eq:risk_aware_sweep_rule} to test the mechanism while keeping the main paper focused on residual-quality hit-rate targeting.

\subsection{Numerical evidence}

Figure~\ref{fig:risk_aware_frontier} compares four policies: raw hit-rate targeting, residual-quality targeting with alpha, residual-quality targeting with naive Sweep participation, and residual-quality targeting with risk-aware style participation. The Sweep fills are random partial fills. The risk-aware rule shifts the attained-service frontier upward relative to the residual-quality baseline and to naive Sweep participation.

\begin{figure}[H]
\centering
\includegraphics[width=0.78\textwidth]{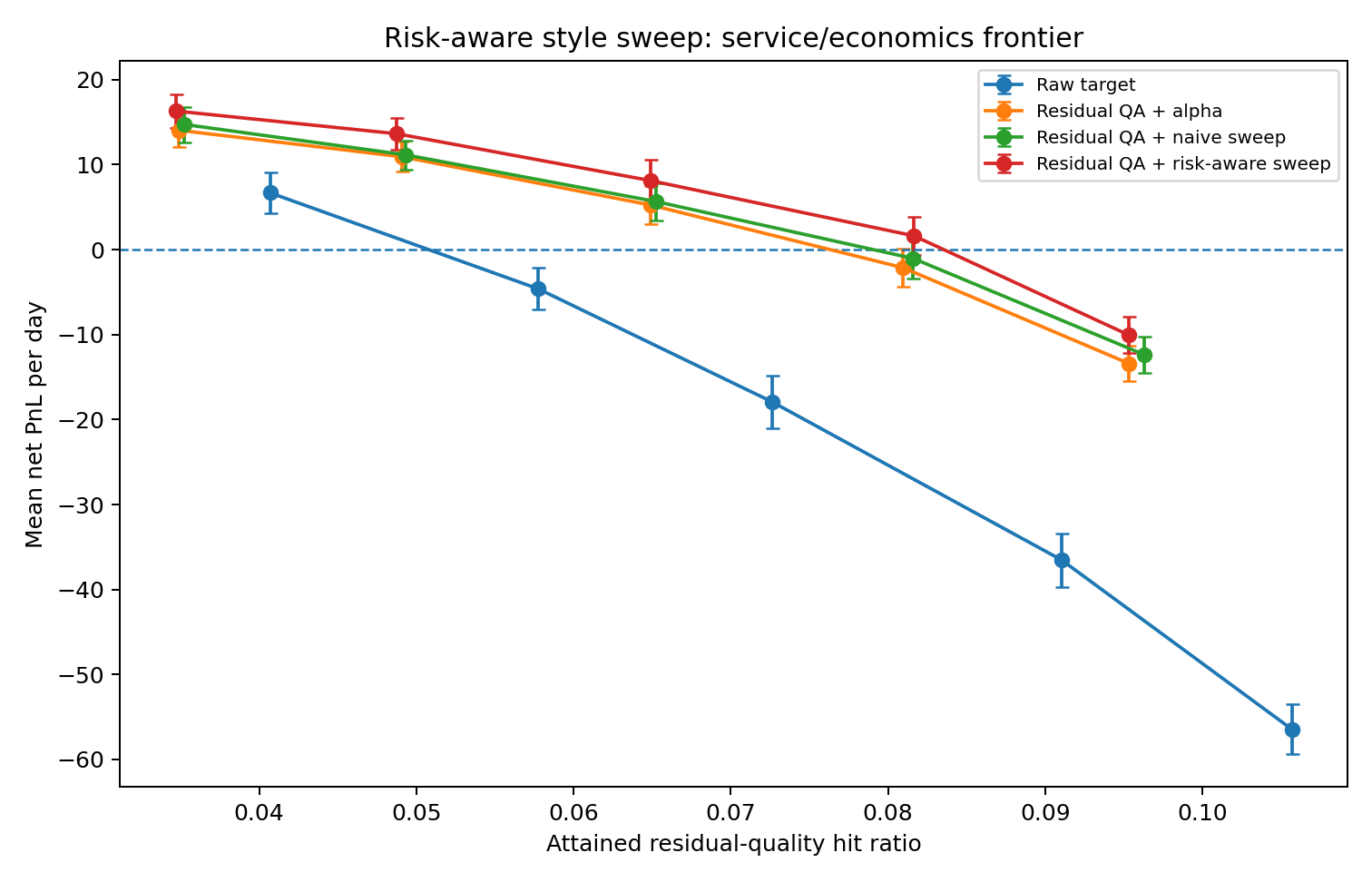}
\caption{Style-aware client-flow warehousing with random Sweep fills. The risk-aware rule improves the service/economics frontier relative to both the residual-quality baseline and naive Sweep participation.}
\label{fig:risk_aware_frontier}
\end{figure}

Table~\ref{tab:risk_aware_diagnostics} reports diagnostics at the 8\% target. The risk-aware policy earns higher PnL than the residual-quality baseline, maintains a similar quality hit ratio, captures more right-way Sweep flow, and reduces terminal inventory risk. The result is not driven by indiscriminately taking more inventory; the participation rule improves the quality of inventory acquired.

\begin{table}[H]
\centering
\caption{Risk-aware style-flow diagnostics at the 8\% target. PnL is bp$\cdot$MM per day.}
\label{tab:risk_aware_diagnostics}
\small
\begin{tabular}{lrrrrrr}
\toprule
Policy & PnL/day & Quality HR & Sweep/day & Right-way & Alpha exposure & Inv. risk \\
\midrule
Raw target & -18.37 & 7.29\% & 0.00 & 0.0\% & -63.8 & 28.4 \\
Residual QA + alpha & 5.85 & 6.44\% & 0.00 & 0.0\% & -25.2 & 20.9 \\
Residual QA + naive sweep & 6.28 & 6.46\% & 6.49 & 59.4\% & -22.7 & 20.6 \\
Residual QA + fixed style sweep & 7.22 & 6.49\% & 6.14 & 77.0\% & -17.2 & 20.8 \\
Residual QA + risk-aware sweep & 8.47 & 6.45\% & 5.96 & 81.9\% & -14.2 & 19.3 \\
\bottomrule
\end{tabular}
\end{table}

Figure~\ref{fig:risk_aware_matched_saving} compares policies at matched attained quality hit ratio. The risk-aware rule improves PnL throughout the overlap region. This supports the interpretation that a high service objective can be more sustainable when marginal client inventory is low-toxicity, style-aligned, and sized by inventory risk.

\begin{figure}[H]
\centering
\includegraphics[width=0.76\textwidth]{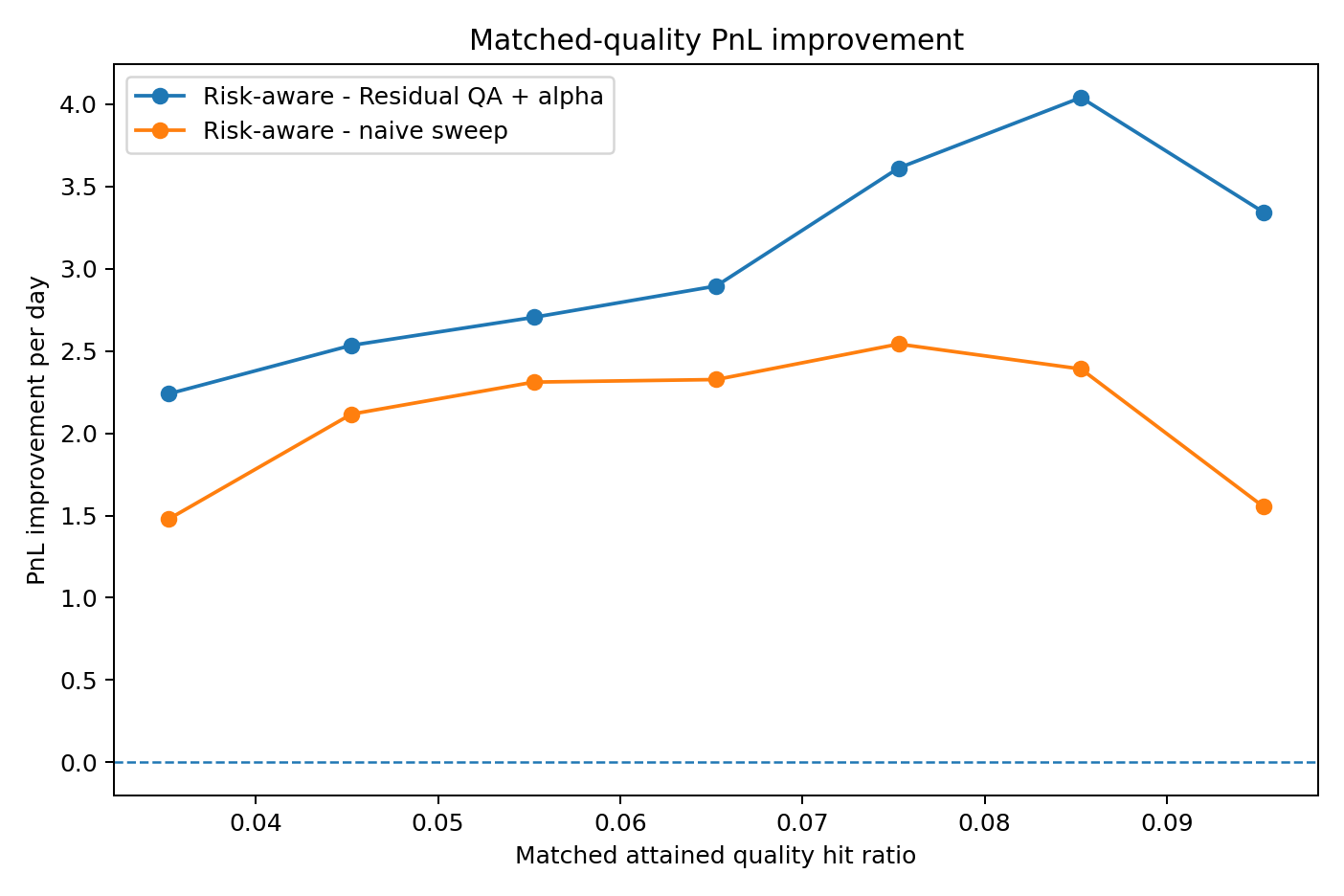}
\caption{Matched-quality comparison for style-aware client-flow warehousing. The risk-aware rule outperforms both the no-Sweep residual-quality baseline and naive Sweep participation.}
\label{fig:risk_aware_matched_saving}
\end{figure}

Table~\ref{tab:risk_aware_matched} reports the matched-quality comparison numerically.

\begin{table}[H]
\centering
\caption{Matched attained-quality comparison. PnL is bp$\cdot$MM per day.}
\label{tab:risk_aware_matched}
\small
\begin{tabular}{rrrrr}
\toprule
Quality HR & Residual QA & Naive Sweep & Risk-aware Sweep & Risk-aware - Residual QA \\
\midrule
3.52\% & 13.98 & 14.74 & 16.22 & 2.24 \\
4.52\% & 11.78 & 12.19 & 14.31 & 2.53 \\
5.53\% & 8.72 & 9.11 & 11.42 & 2.70 \\
6.53\% & 5.09 & 5.65 & 7.98 & 2.89 \\
7.53\% & 0.48 & 1.55 & 4.09 & 3.61 \\
8.53\% & -5.53 & -3.88 & -1.49 & 4.04 \\
9.53\% & -13.39 & -11.60 & -10.05 & 3.34 \\
\bottomrule
\end{tabular}
\end{table}

The PnL attribution in Table~\ref{tab:risk_aware_attr} shows where the improvement comes from. Risk-aware participation adds Sweep spread capture, improves style-alpha PnL, and reduces terminal liquidation cost. It also incurs Sweep residual and acquisition costs, but these are more than offset in the synthetic experiment.

\begin{table}[H]
\centering
\caption{Risk-aware Sweep minus residual-quality baseline at the 8\% target. Positive values improve PnL.}
\label{tab:risk_aware_attr}
\small
\begin{tabular}{lr}
\toprule
Component & Contribution \\
\midrule
RFQ spread capture & 1.15 \\
Sweep spread capture & 1.19 \\
Style-alpha PnL & 0.59 \\
RFQ residual toxicity cost saved & 0.24 \\
Sweep residual cost saved & -0.48 \\
Terminal liquidation cost saved & 0.22 \\
Net PnL improvement & 2.62 \\
\bottomrule
\end{tabular}
\end{table}

The conclusion is deliberately modest. Style alignment alone is not sufficient; a fixed participation rule can over-warehouse risk. However, a participation rule derived from the same quadratic value approximation as the RFQ quote control can convert some client service into alpha-aligned inventory acquisition while controlling inventory risk and random fill uncertainty.

\section{Calibration roadmap without proprietary data}

This paper is currently using synthetic data. A practical calibration would require:
\begin{enumerate}[leftmargin=*]
    \item RFQ arrival intensities $\lambda_{m,\tau,s,k}$ by bond, segment, side, and size.
    \item Fill curves $f_{m,\tau,s,k}(\delta)$ estimated from quote distance and outcomes.
    \item Credit covariance $\Sigma=B\Omega B^\top+\Sigma_\epsilon$ from spread or price returns.
    \item Credit alpha $\mu_m$ from carry, rolldown, issuer-RV, flow, and index signals.
    \item Residual toxicity $\chi^{res}_{m,\tau,s,k}$ from factor- and alpha-adjusted dealer-signed markouts.
\end{enumerate}

The method does not necessarily require a model per client but can be based aggregate archetypal tiers, such as carry/RV (tier 1 or 2), passive/index (tier 1) and residual-toxic flow (tier 3) segments for regulatory purposes.

\section{Limitations and validation requirements}

The model deliberately separates mathematical tractability from empirical calibration. This has three consequences.

First, the simulations are synthetic. They show that residual-quality targeting can dominate raw hit-ratio targeting under plausible flow-quality heterogeneity, but they do not estimate the magnitude of the effect on any real desk.

Second, residual toxicity is only as good as the markout adjustment. If carry, rolldown, issuer-RV, index demand, or factor beta are omitted from the adjustment, the model may misclassify sophisticated but non-toxic clients as toxic. This is precisely why the residual definition \eqref{eq:markout_decomp} is central.

Third, the style-aware warehousing extension and the passive/index special case are reduced-form. They illustrate inventory-recycling value but do not yet solve a full joint optimal acquisition-and-quoting control problem. The style-aware extension uses a quadratic value-based participation rule with random fills; it should be read as a mechanism study rather than as a calibrated empirical estimate.

In an empirical application, the headline frontier would be calibrated with RFQ arrival rates, fill curves, credit factor exposures, and factor-adjusted markout estimates from observed trading data. The present paper should therefore be read as a theoretical and numerical mechanism study.

\section{Conclusion}

Hit ratio is not wrong. Uniform raw hit ratio is wrong. In corporate bond RFQ market making, the dealer should distinguish between flow that is adverse because it contains residual private information and flow that looks adverse only because the dealer failed to price observable credit alpha, carry/rolldown, issuer-RV, or index-demand effects.

A residual-quality-adjusted hit-ratio target preserves the tractability of the hit-ratio HJB while changing the economic objective. The exact-dual Hamiltonian remains separable, the quote map remains explicit for logistic fills, and the multi-bond approximation remains scalable through the Riccati matrix. The result is a desk-interpretable quoting rule: subsidize service selectively for low-residual-toxicity, recyclable, and forecastable flow; do not subsidize residual-toxic flow merely to satisfy a raw hit-ratio KPI. A broader extension is style-aware client-flow warehousing: when client inventory opportunities are aligned with the dealer's risk-adjusted credit alpha, they can turn client service into desired inventory acquisition, but only if participation is sized for execution uncertainty and inventory risk. The reduced-form risk-aware participation experiment supports this mechanism while keeping the main contribution focused on residual-quality hit-rate targeting.

\appendix

\section{Exact scalar dual solve}

For fixed inventory state and tier, define
\begin{equation}
    F_\tau(\xi)=
    \xi-
    \kappa_\tau
    \left[
    r^{Q,\star}_\tau+
    \frac{1}{W^Q_\tau}
    \sum_{i\in\mathcal I_\tau} z_i w_i H_i'(p_i^0-w_i\xi)
    \right].
\end{equation}
Since $H_i''>0$ and $w_i\ge0$,
\begin{equation}
    F_\tau'(\xi)=
    1+
    \frac{\kappa_\tau}{W^Q_\tau}
    \sum_{i\in\mathcal I_\tau} z_i w_i^2 H_i''(p_i^0-w_i\xi)>0.
\end{equation}
Therefore $F_\tau$ is strictly increasing and has at most one root. In implementation, one brackets the root and applies bisection or Brent's method. A fast approximation used during exploratory sweeps is
\begin{equation}
    \xi\approx
    \frac{\kappa_\tau(r^{Q,\star}_\tau-r^Q_{\tau,0})}
    {1+\kappa_\tau J_0},
    \qquad
    J_0=\frac{1}{W^Q_\tau}\sum_i z_i w_i^2 H_i''(p_i^0),
\end{equation}
but the headline algorithm in this paper uses the exact scalar solve.

\section{Numerical precision and replication}
\label{app:replication_plan}

The numerical results can be reproduced and sharpened by rerunning the nonlinear-dual Monte Carlo with a larger path count. A higher-precision configuration is:
\begin{enumerate}[leftmargin=*]
    \item Use the same synthetic universe and opportunity grid for reproducibility.
    \item Use nonlinear scalar dual solves for $\xi_\tau$ at every quote state.
    \item Run at least $1{,}000$ Monte Carlo paths per policy and target for the main frontier.
    \item Report standard errors or bootstrap confidence intervals for PnL/day, quality hit ratio, toxic hit ratio, and residual cost/day.
    \item Keep the linearized dual closure only as an appendix speed approximation, not as the headline result.
\end{enumerate}

A high-precision replication should report
\[
    \Delta \mathrm{PnL}(r^\star),\quad
    \Delta r^{toxic}(r^\star),\quad
    \Delta \text{ResidualCost}(r^\star)
\]
for each target $r^\star\in\{4\%,6\%,8\%,10\%,12\%\}$.

\section{Reduced-form passive/index experiment details}
\label{app:passive_details}

This appendix summarizes a reduced-form passive/index experiment that can be viewed as a special case of forecastable, low-residual-toxicity client flow. The synthetic passive signal assigns each bond a signed event demand $I_m$ based on maturity, rating, sector, credit alpha, and idiosyncratic noise. Positive $I_m$ corresponds to passive clients buying from the dealer. Event demand arrives over a concentrated window, shown in Figure~\ref{fig:passive_event_profile}.

\begin{figure}[H]
\centering
\includegraphics[width=0.72\textwidth]{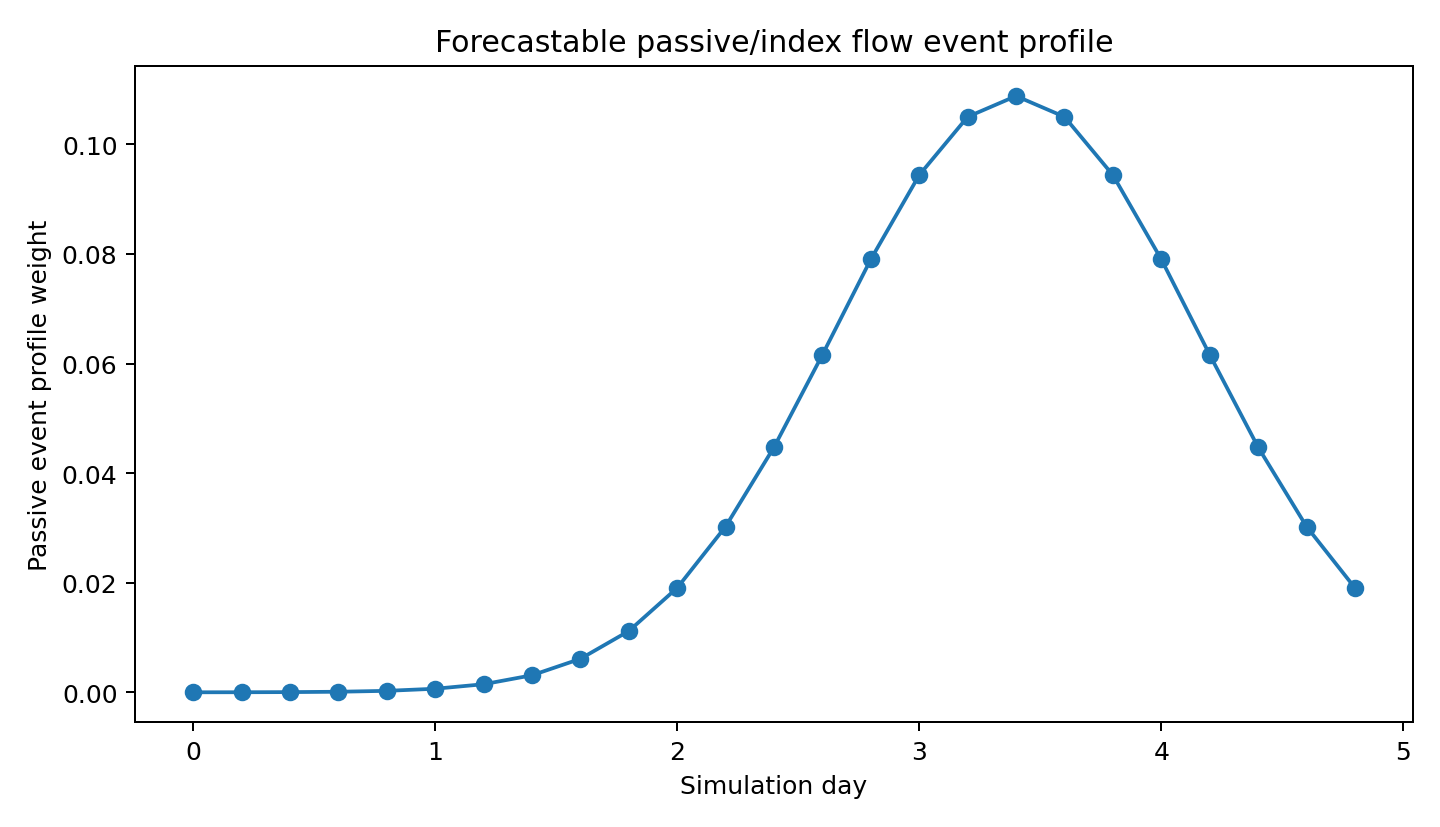}
\caption{Synthetic passive/index event profile used in the reduced-form extension.}
\label{fig:passive_event_profile}
\end{figure}

The sweep grid varies forecast skill $\rho$ and pre-position fraction $\theta$. Figure~\ref{fig:sweep_heatmap} shows that aggressive pre-positioning is costly when forecasts are weak, while moderate pre-positioning becomes attractive only when forecast quality is high.

\begin{figure}[H]
\centering
\includegraphics[width=0.72\textwidth]{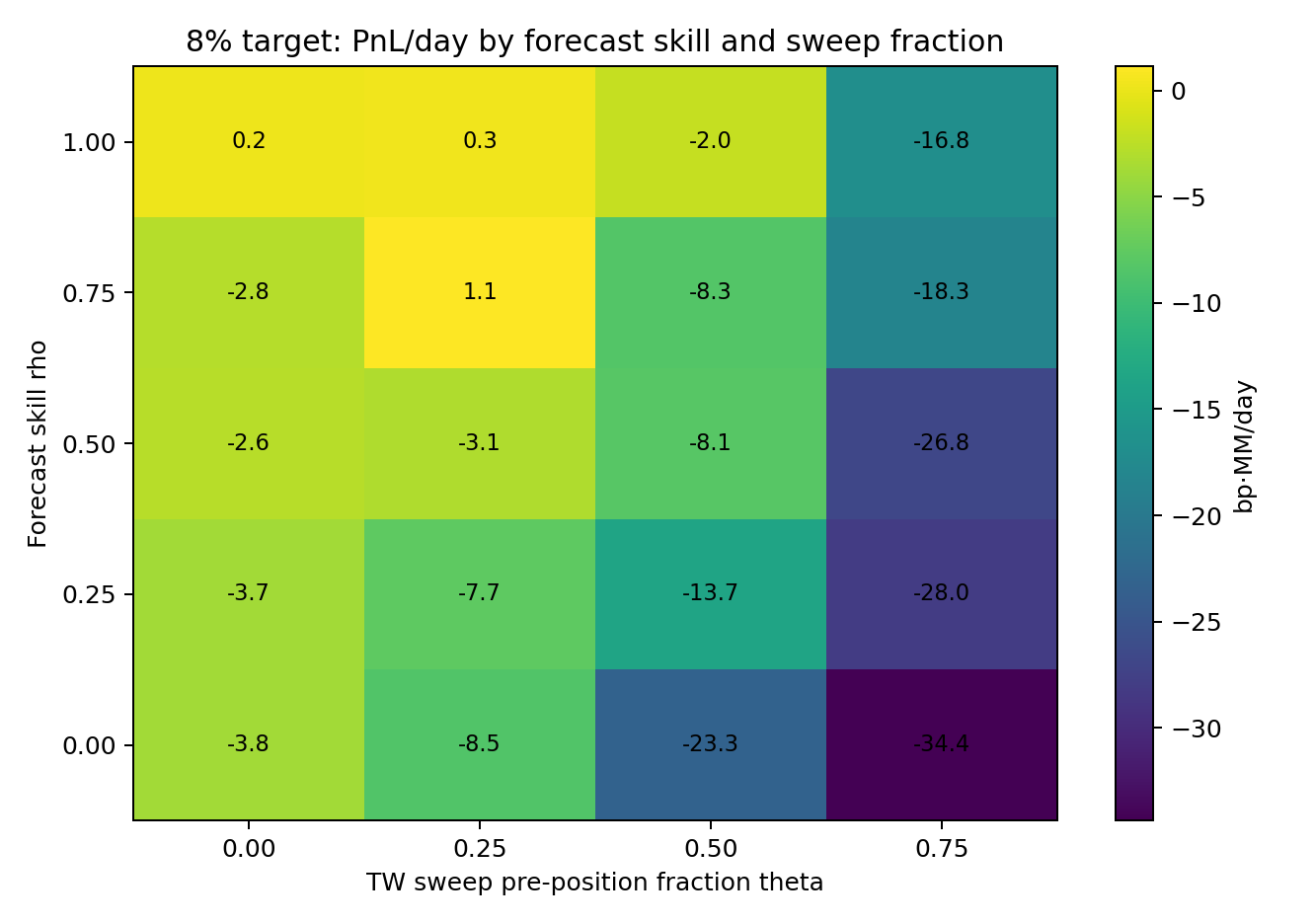}
\caption{TW Sweep scenario grid at an 8\% target. Cells report mean PnL per day.}
\label{fig:sweep_heatmap}
\end{figure}

Figure~\ref{fig:optimal_theta} reports the best sweep fraction in the finite grid. The optimal scenario is no pre-positioning for low forecast skill and modest pre-positioning for high forecast skill.

\begin{figure}[H]
\centering
\includegraphics[width=0.72\textwidth]{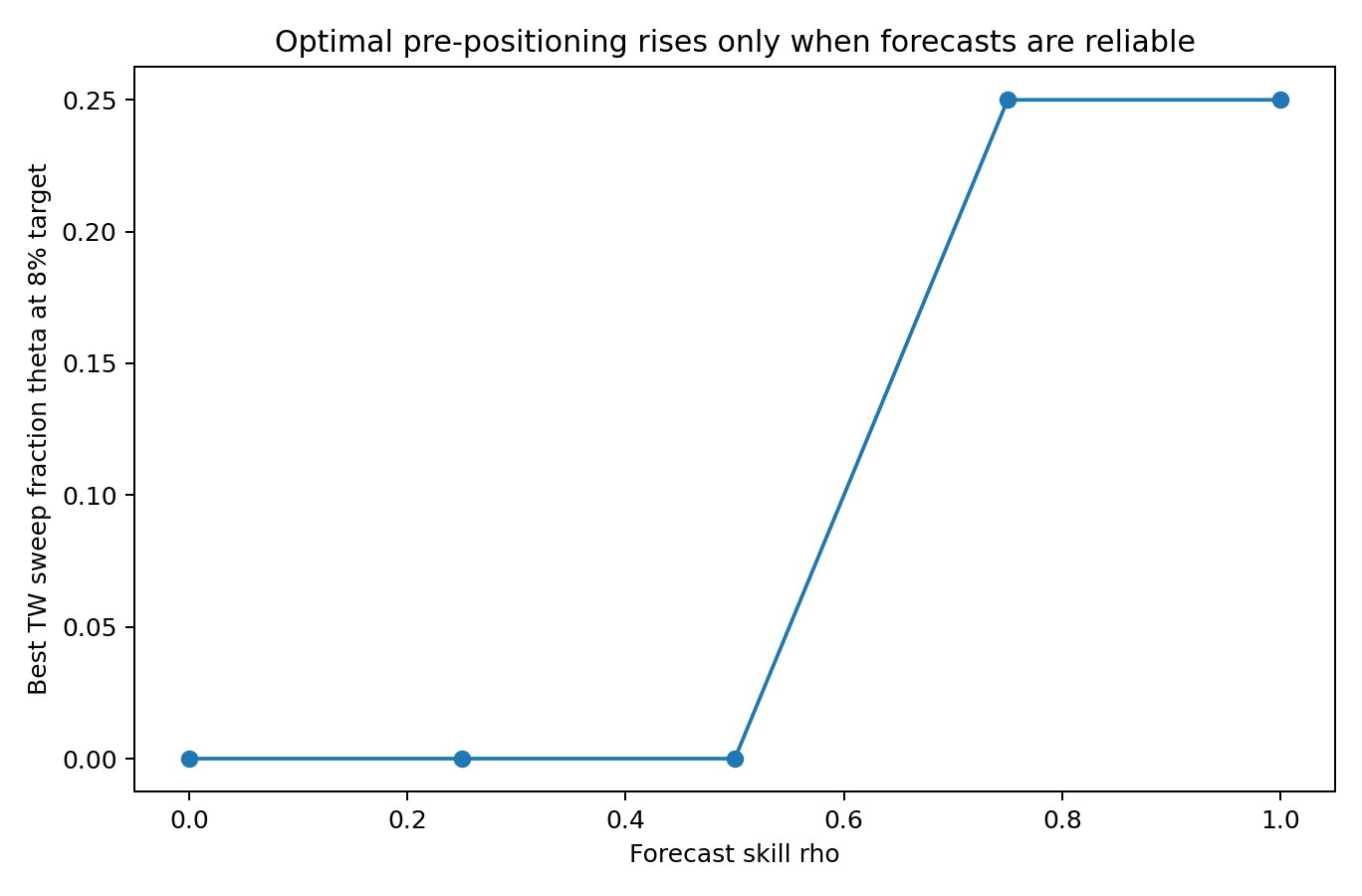}
\caption{Best TW Sweep fraction in the reduced-form grid. Forecast quality, not raw passive demand, determines whether pre-positioning is useful.}
\label{fig:optimal_theta}
\end{figure}

The supplementary files include the passive demand signal, event profile, RFQ opportunity grid with the passive segment, sustainable-hit-rate table, sweep grid, and optimal-sweep table. These are synthetic mechanism checks and should not be read as empirical estimates of passive-flow profitability for any particular desk.

\clearpage
\section{Supplementary numerical figure appendix}
\label{app:all_figures}

This appendix reports supplementary numerical diagnostics for the synthetic experiments. The main text focuses on headline figures and tables; the figures below document robustness checks, implementation checks, and lower-dimensional examples. Some supplementary figures use earlier local-dual approximations or smaller nonlinear-dual runs, as stated in the captions; the main text reports the nonlinear-dual attained-quality frontier.

\subsection{Baseline one-bond diagnostics}

\begin{figure}[H]
\centering
\includegraphics[width=0.76\textwidth]{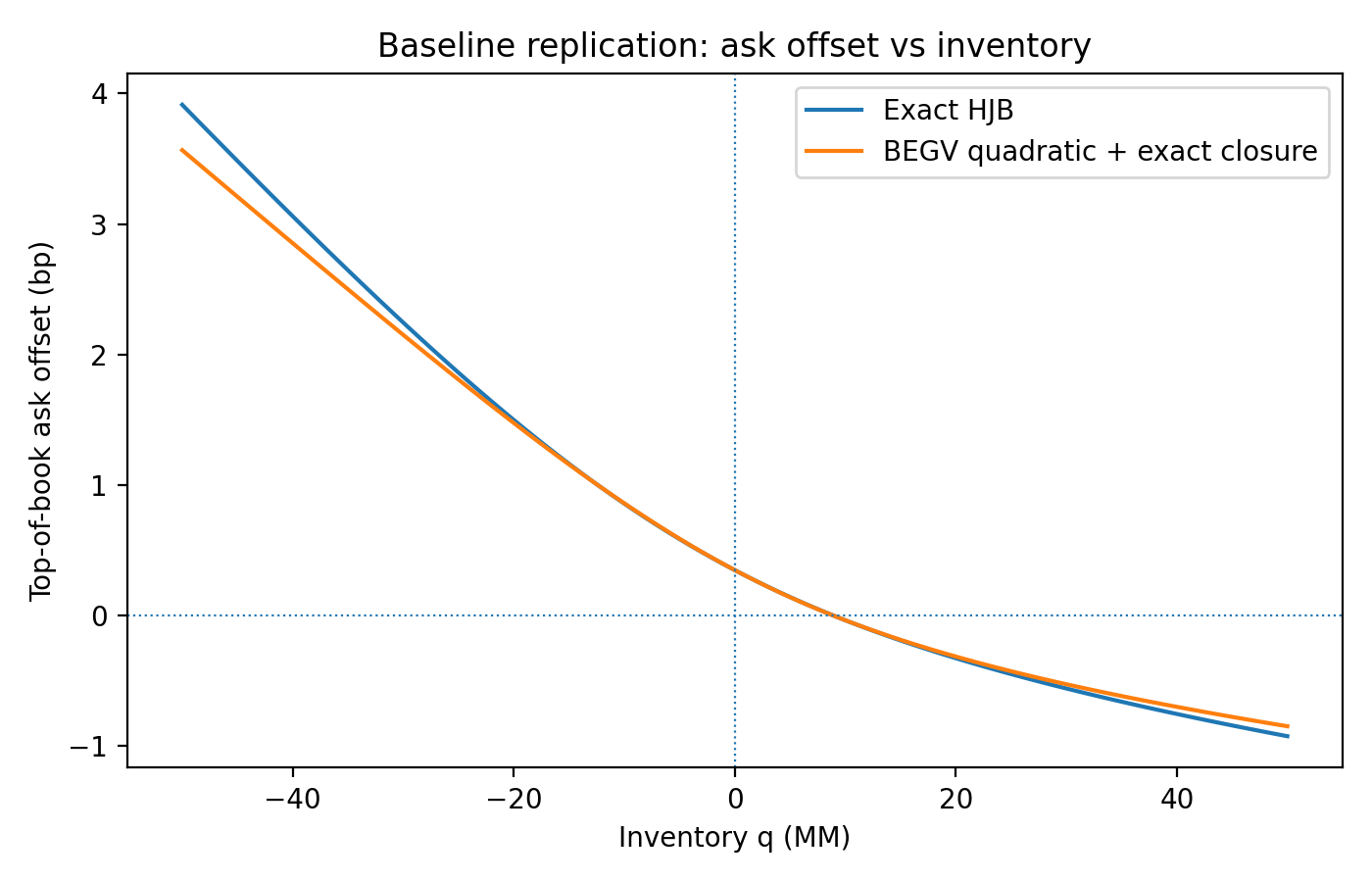}
\caption{Baseline one-bond ask quote offset as a function of inventory. The figure checks the expected inventory skew: a long dealer tightens offers and widens bids, while a short dealer does the opposite.}
\label{fig:app_step1_ask}
\end{figure}

\begin{figure}[H]
\centering
\includegraphics[width=0.76\textwidth]{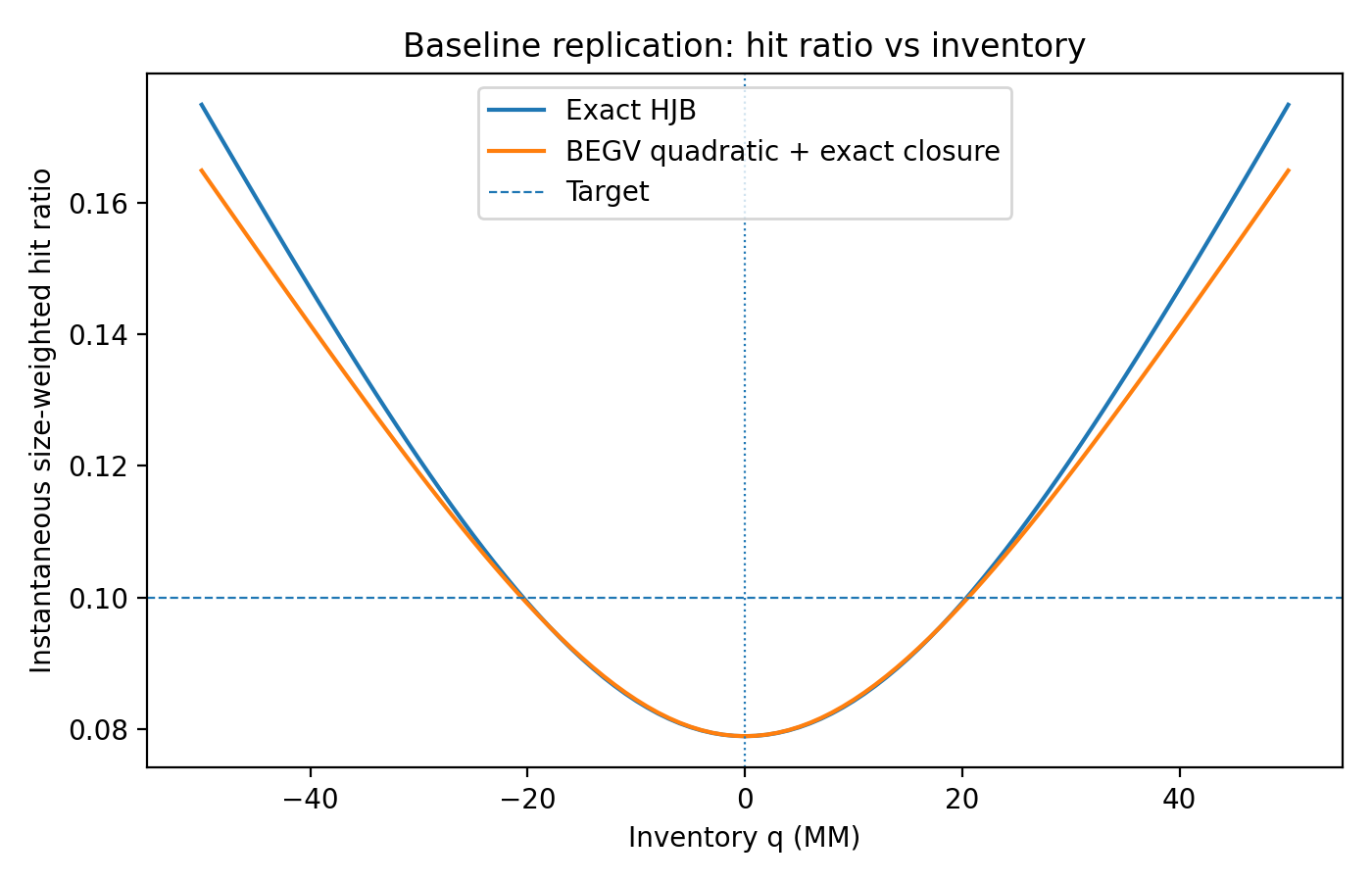}
\caption{Baseline one-bond hit ratio as a function of inventory. The hit-ratio target interacts with inventory control: service subsidy is strongest near inventory states where risk constraints are least binding.}
\label{fig:app_step1_hr}
\end{figure}

\begin{figure}[H]
\centering
\includegraphics[width=0.76\textwidth]{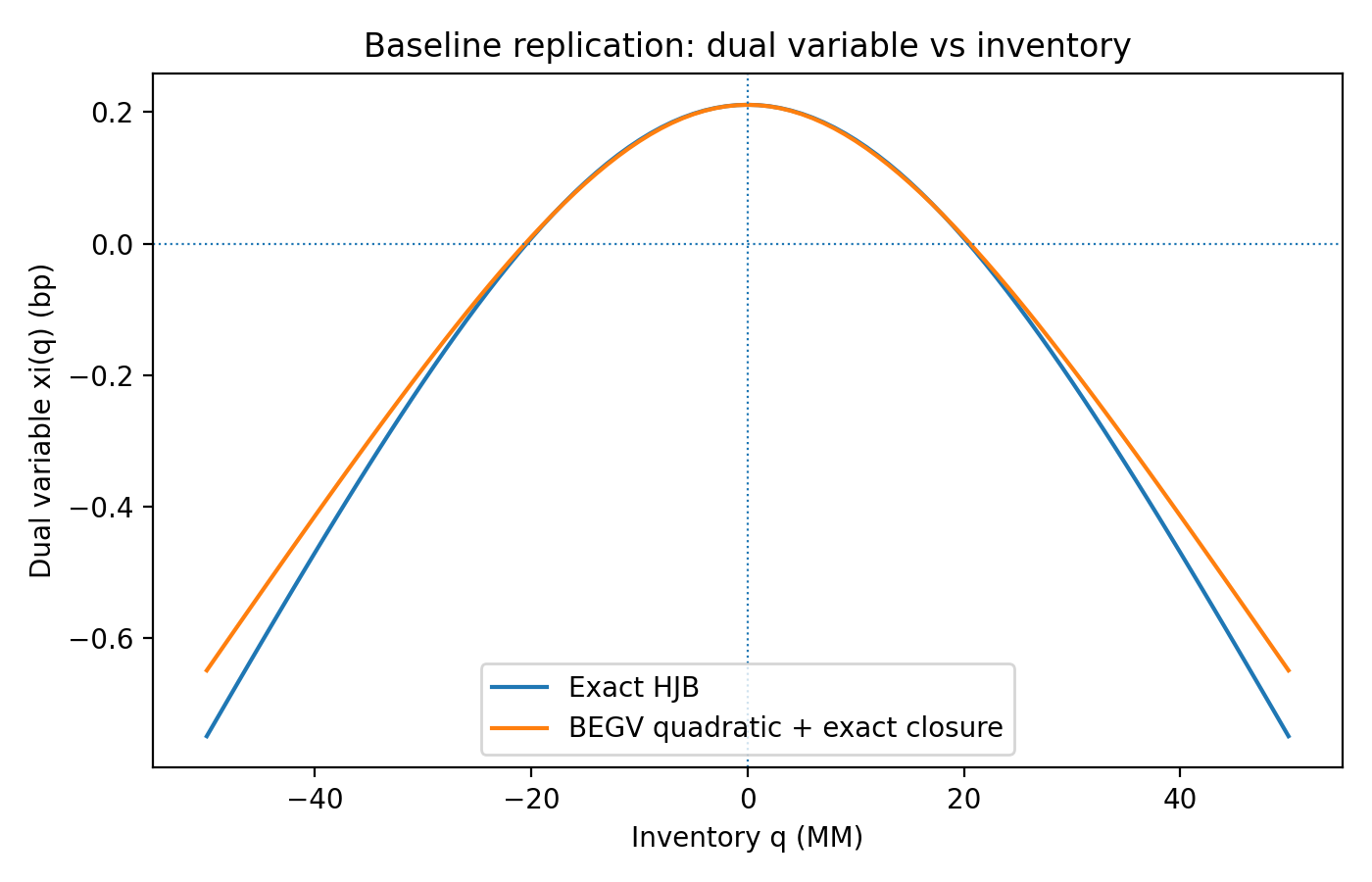}
\caption{Baseline one-bond dual variable as a function of inventory. The nonlinearity of the dual variable motivates solving the scalar fixed point rather than relying on a global linear approximation.}
\label{fig:app_step1_xi}
\end{figure}

\subsection{Residual-toxicity identification diagnostics}

\begin{figure}[H]
\centering
\includegraphics[width=0.76\textwidth]{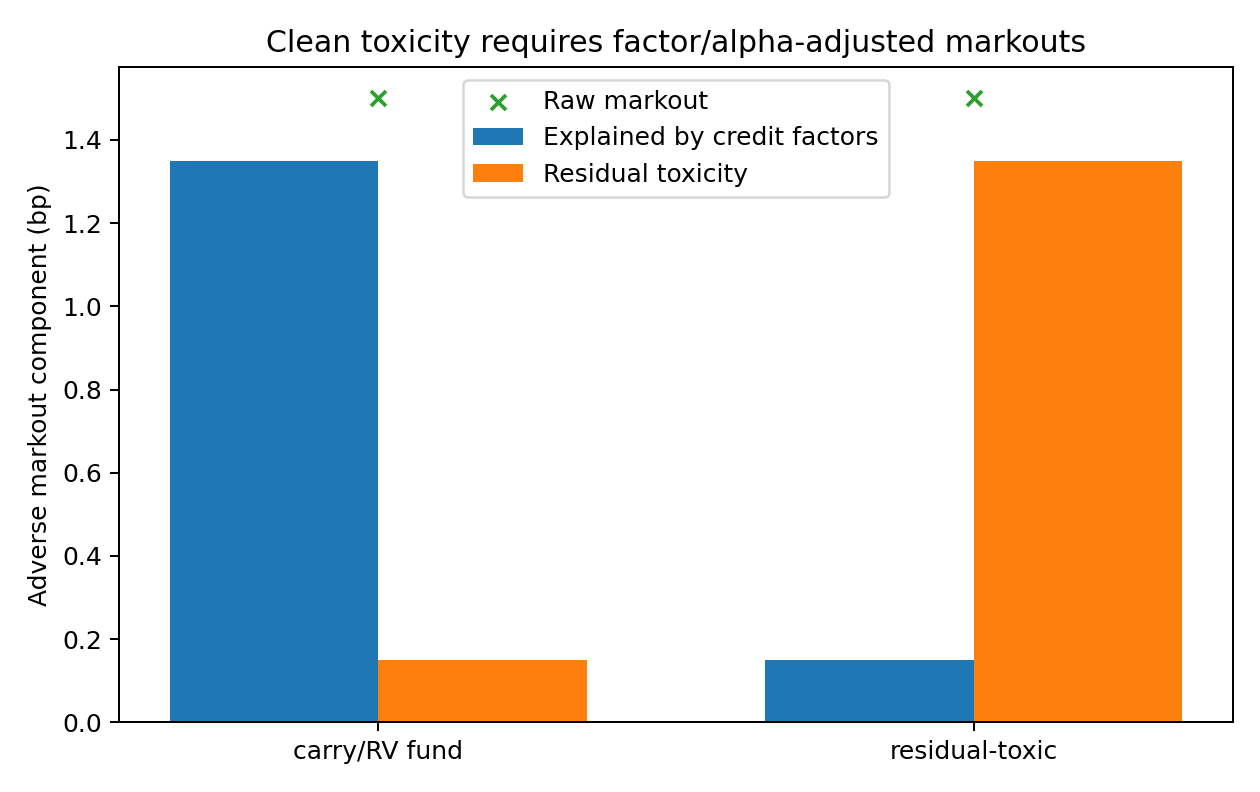}
\caption{Synthetic markout decomposition used to motivate residual toxicity. Two flows can have similar raw adverse markouts but very different residual adverse-selection content after controlling for credit factors and alpha.}
\label{fig:app_step2_decomp}
\end{figure}

\begin{figure}[H]
\centering
\includegraphics[width=0.76\textwidth]{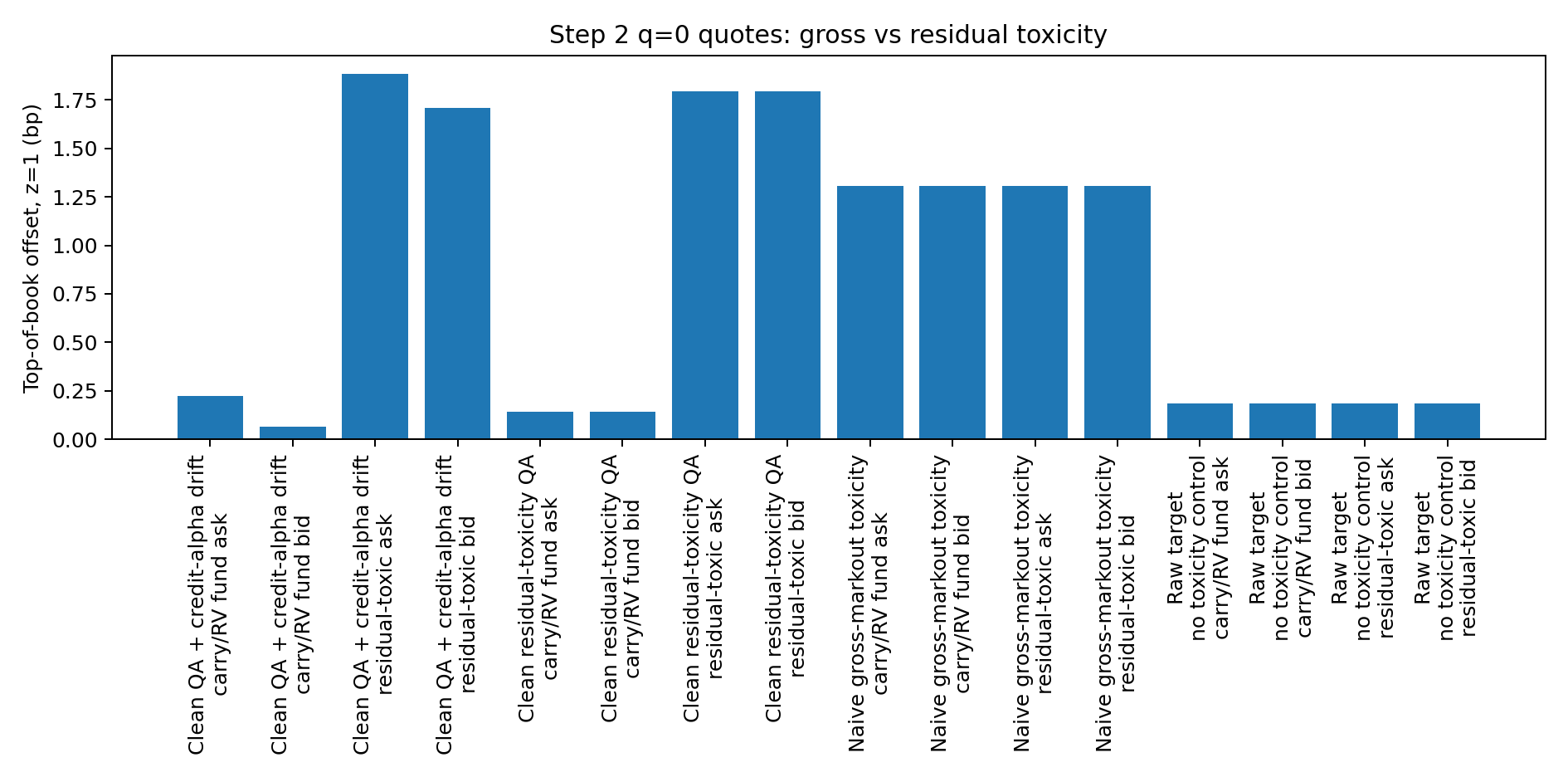}
\caption{Top-of-book quote offsets at zero inventory under raw targeting, naive gross-markout toxicity, and residual-toxicity quality adjustment. Gross-markout toxicity penalizes carry/RV flow that should instead be explained by credit alpha.}
\label{fig:app_step2_q0}
\end{figure}

\begin{figure}[H]
\centering
\includegraphics[width=0.76\textwidth]{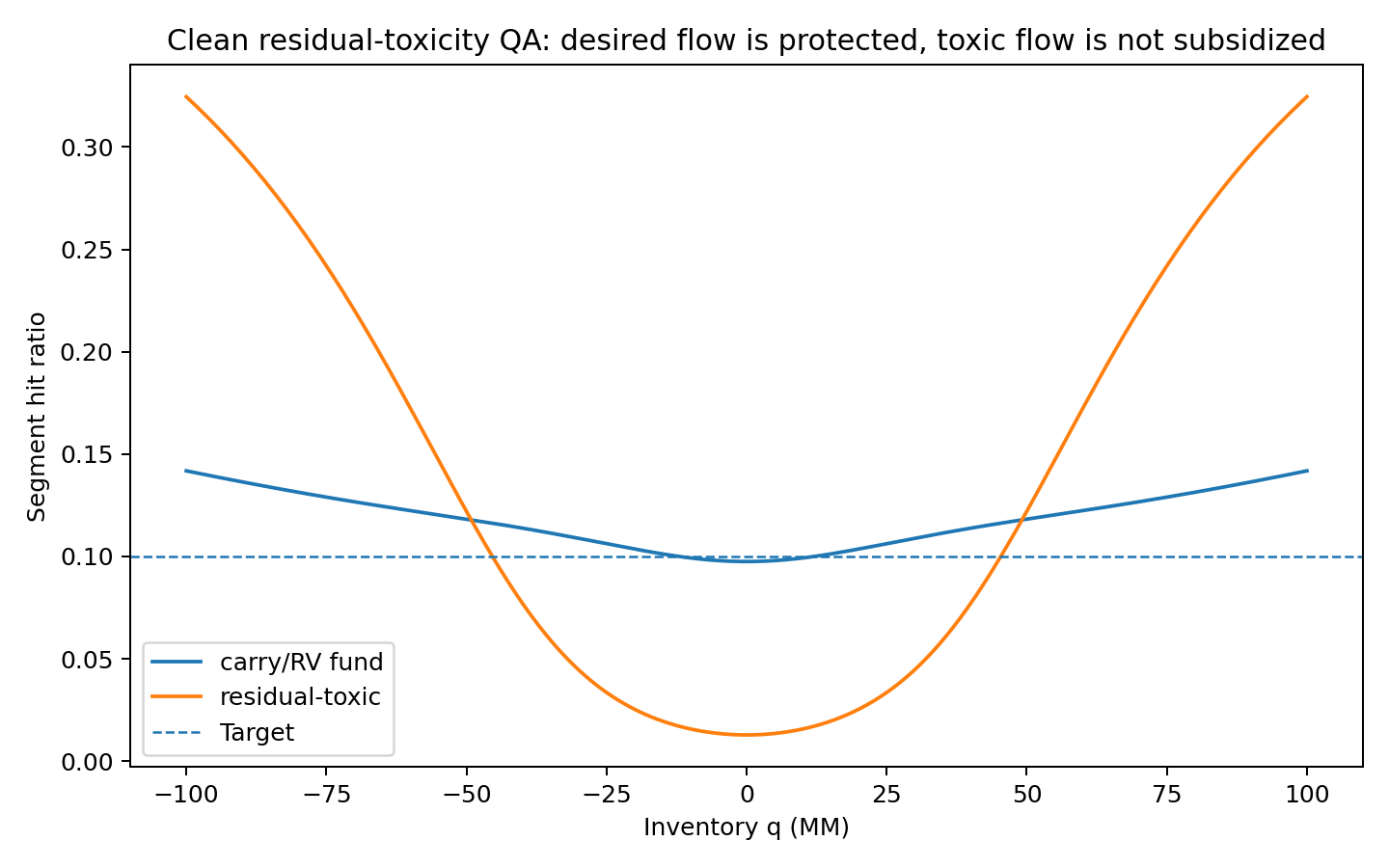}
\caption{Segment hit ratios under the clean residual-toxicity policy as inventory varies. The policy preserves service to low-residual-toxicity flow and avoids subsidizing residual-toxic flow.}
\label{fig:app_step2_clean_hr}
\end{figure}

\begin{figure}[H]
\centering
\includegraphics[width=0.76\textwidth]{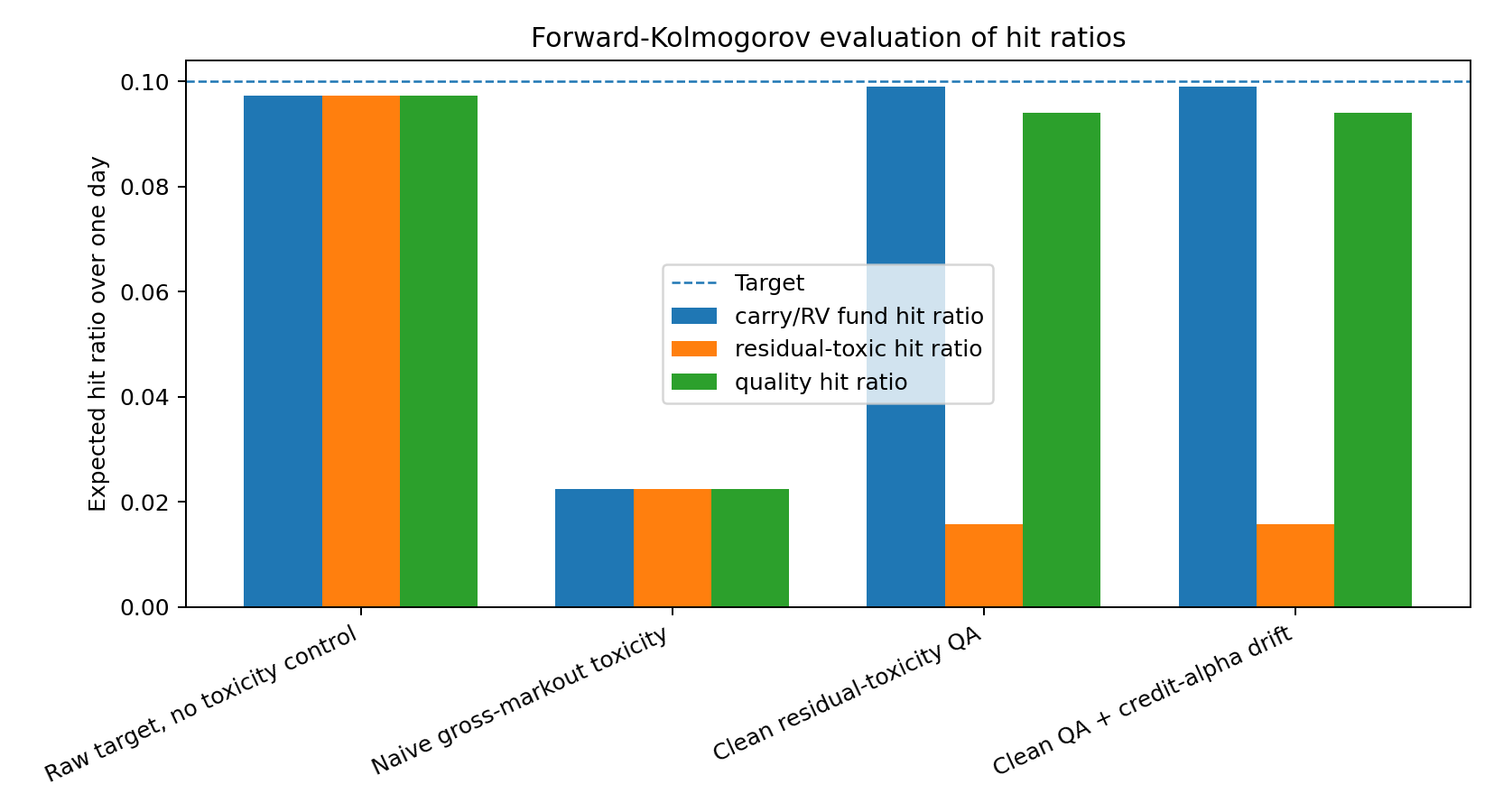}
\caption{Forward-Kolmogorov hit-ratio comparison in the one-bond residual-toxicity example. Residual-quality targeting achieves quality service without raising toxic-flow hit rates.}
\label{fig:app_step2_forward_hr}
\end{figure}

\begin{figure}[H]
\centering
\includegraphics[width=0.76\textwidth]{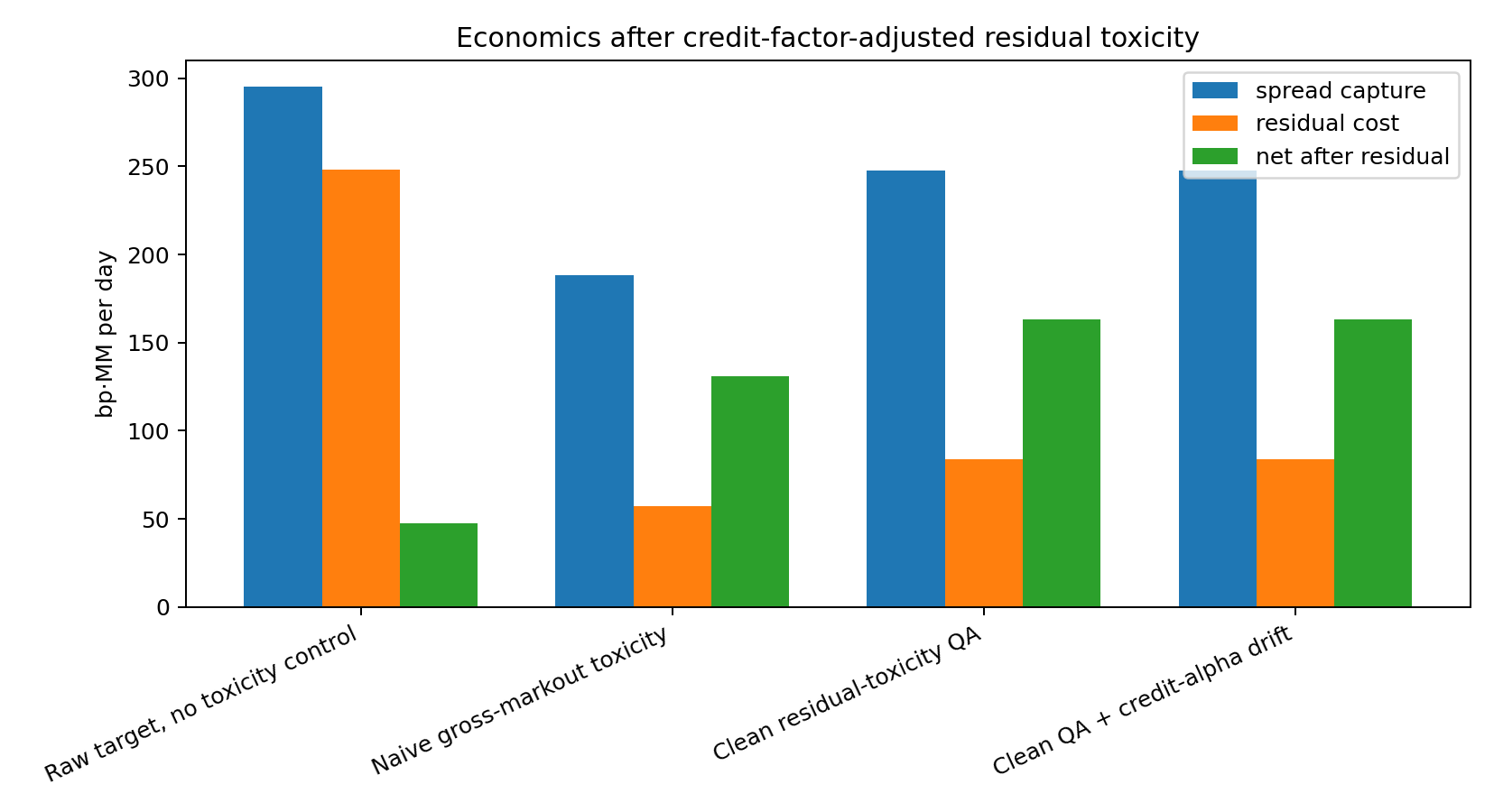}
\caption{Forward-Kolmogorov residual economics in the one-bond example. The clean residual policy reduces residual adverse-selection cost while retaining desirable flow.}
\label{fig:app_step2_forward_econ}
\end{figure}

\subsection{Multi-bond synthetic diagnostics}

\begin{figure}[H]
\centering
\includegraphics[width=0.76\textwidth]{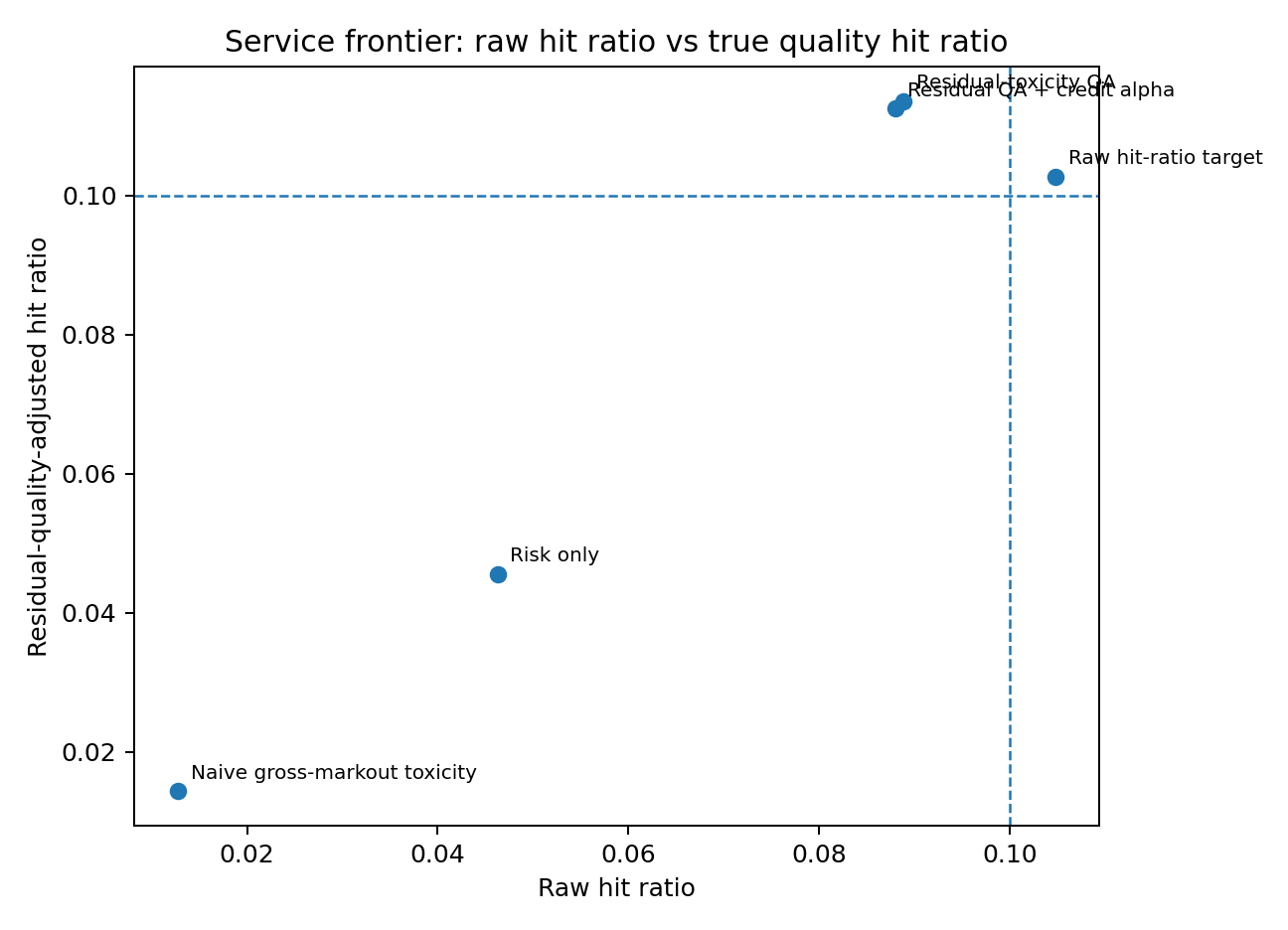}
\caption{Raw hit ratio versus residual-quality hit ratio across policies in the multi-bond synthetic simulation. This figure illustrates why raw service and quality-adjusted service should be distinguished.}
\label{fig:app_step3_raw_quality}
\end{figure}

\begin{figure}[H]
\centering
\includegraphics[width=0.76\textwidth]{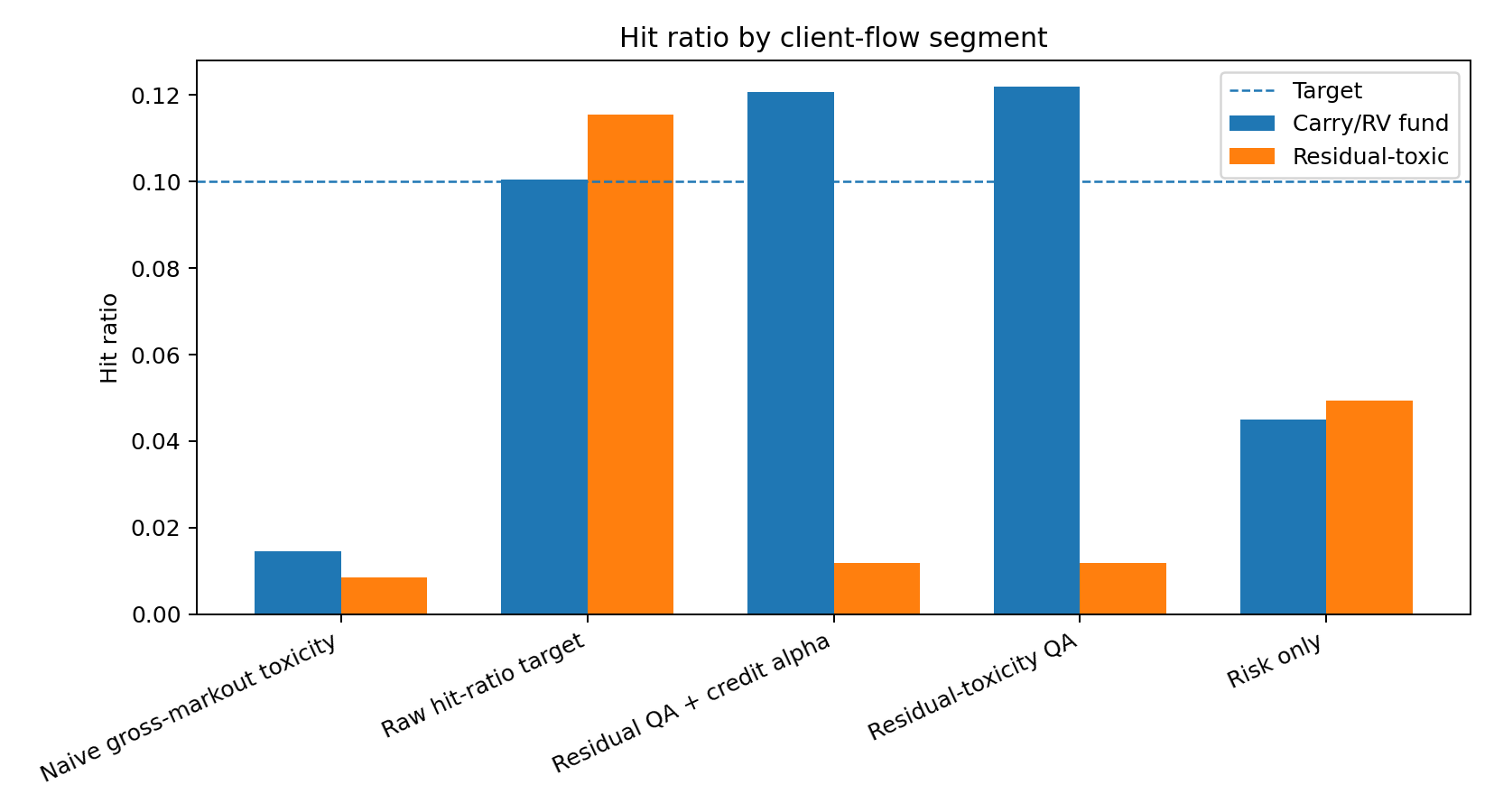}
\caption{Segment hit ratios in the multi-bond simulation. Raw hit-ratio targeting wins residual-toxic flow, while residual-quality targeting reallocates fills toward lower-residual-toxicity flow.}
\label{fig:app_step3_segment}
\end{figure}

\begin{figure}[H]
\centering
\includegraphics[width=0.76\textwidth]{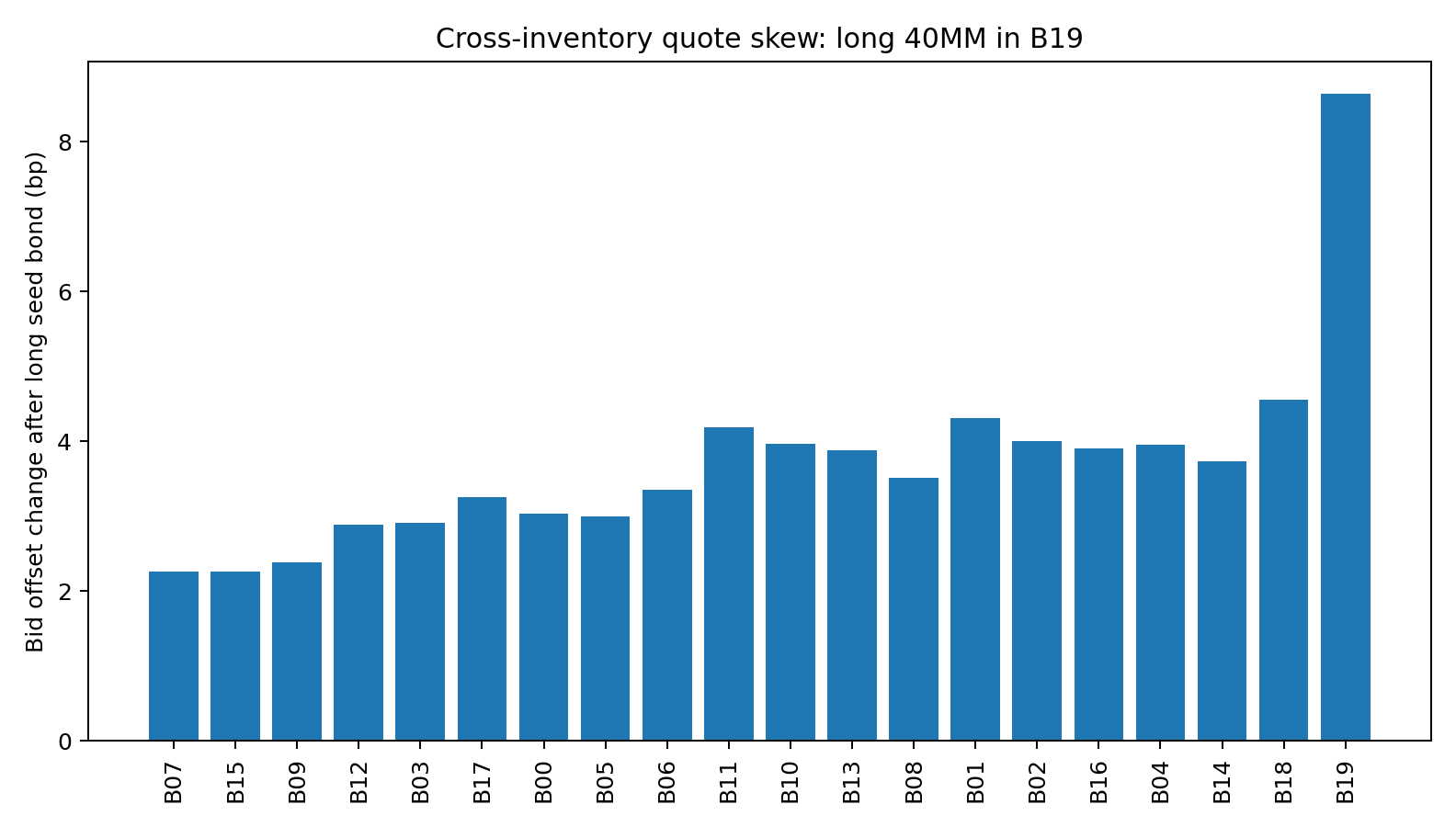}
\caption{Cross-inventory quote skew in a credit-factor book. A long position in one bond widens bids not only in that bond, but also in correlated bonds through the factor covariance matrix.}
\label{fig:app_step3_cross_skew}
\end{figure}

\begin{figure}[H]
\centering
\includegraphics[width=0.76\textwidth]{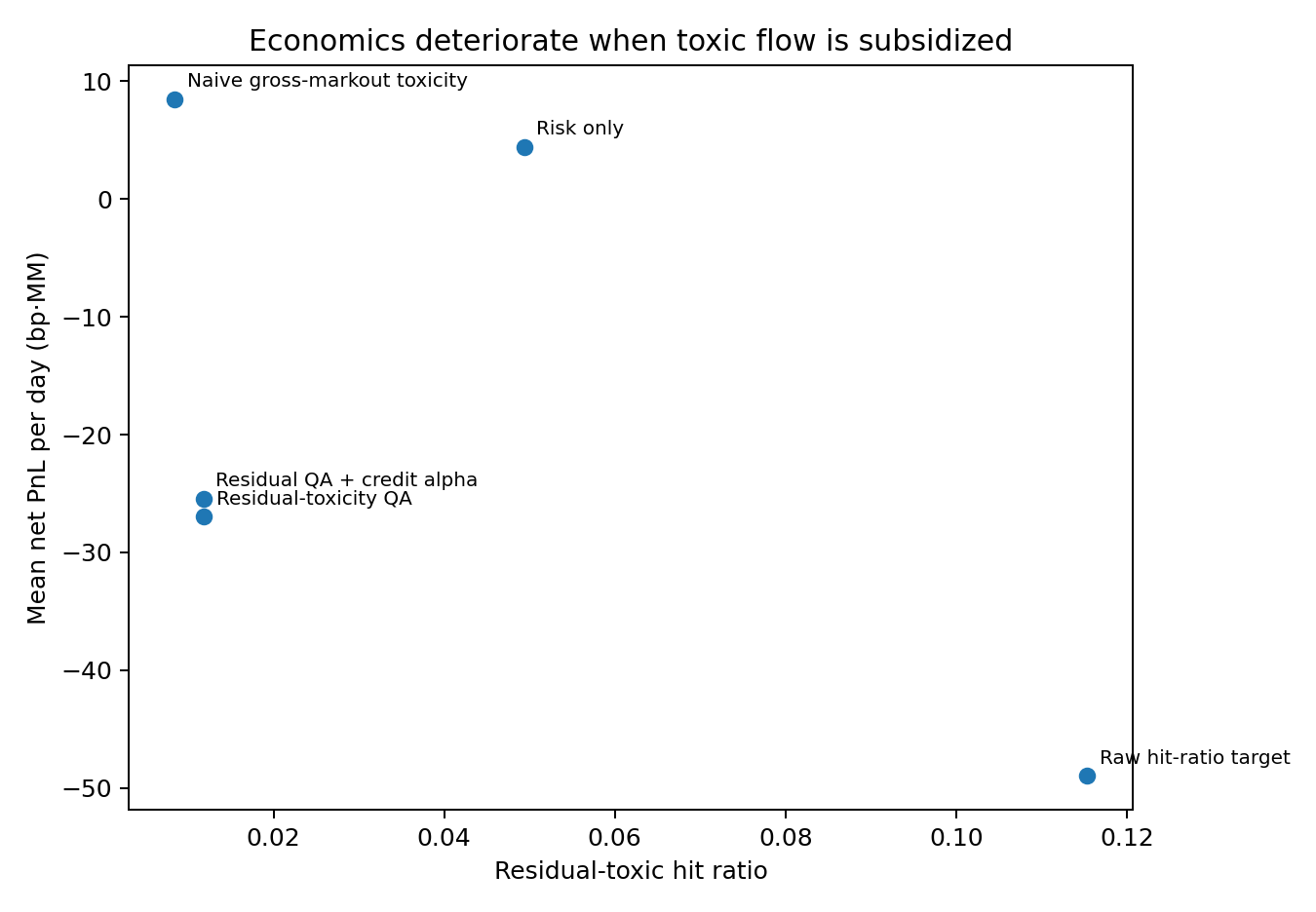}
\caption{Mean PnL versus residual-toxic hit ratio in the multi-bond simulation. Policies that subsidize residual-toxic flow suffer worse economics.}
\label{fig:app_step3_pnl_toxic}
\end{figure}

\begin{figure}[H]
\centering
\includegraphics[width=0.76\textwidth]{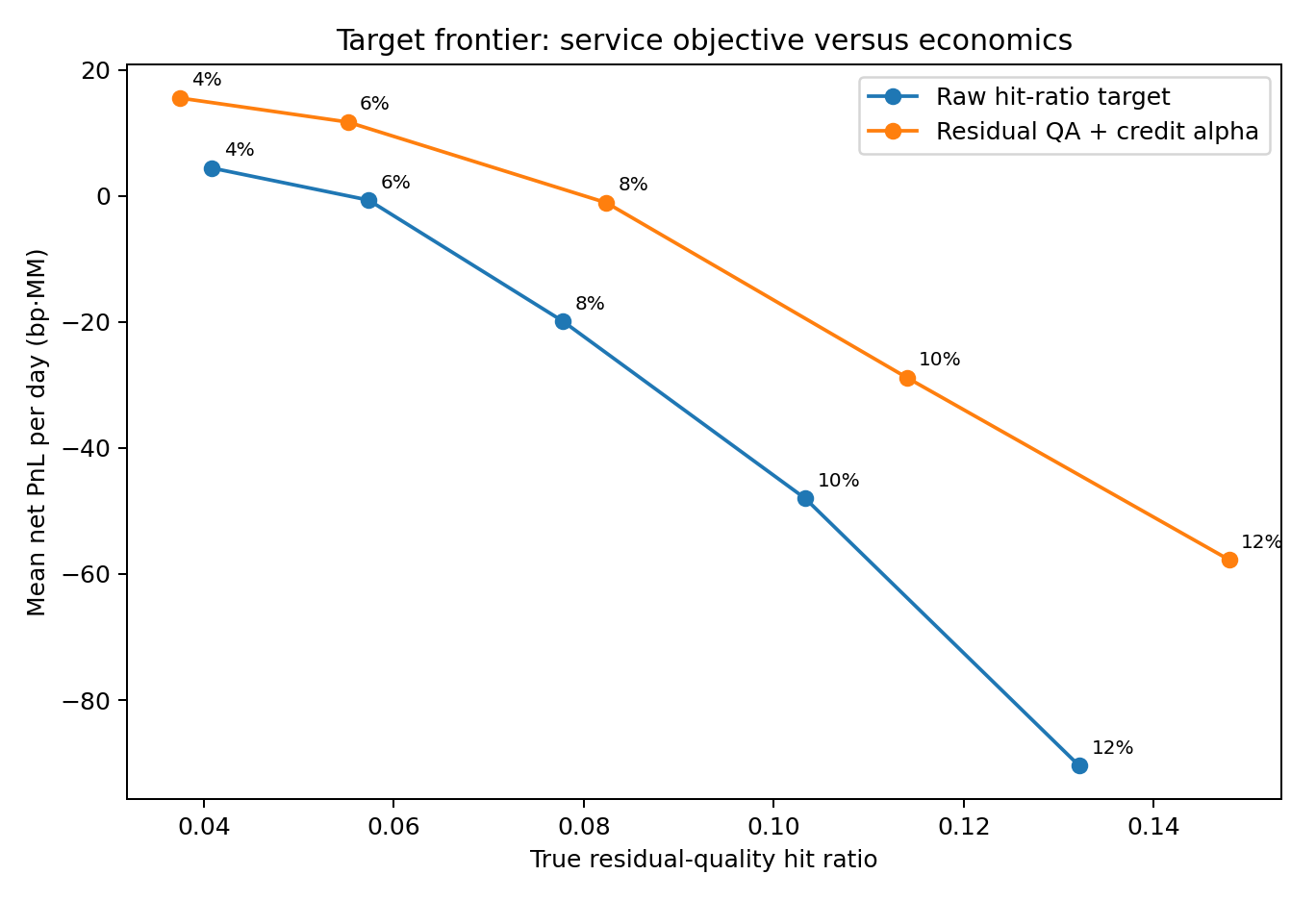}
\caption{Target frontier from the multi-bond simulation before interpolation by attained residual-quality hit ratio. The main text reports the attained-service version of this comparison.}
\label{fig:app_step3_target_frontier}
\end{figure}

\subsection{Robustness and dual-solve diagnostics}

\begin{figure}[H]
\centering
\includegraphics[width=0.76\textwidth]{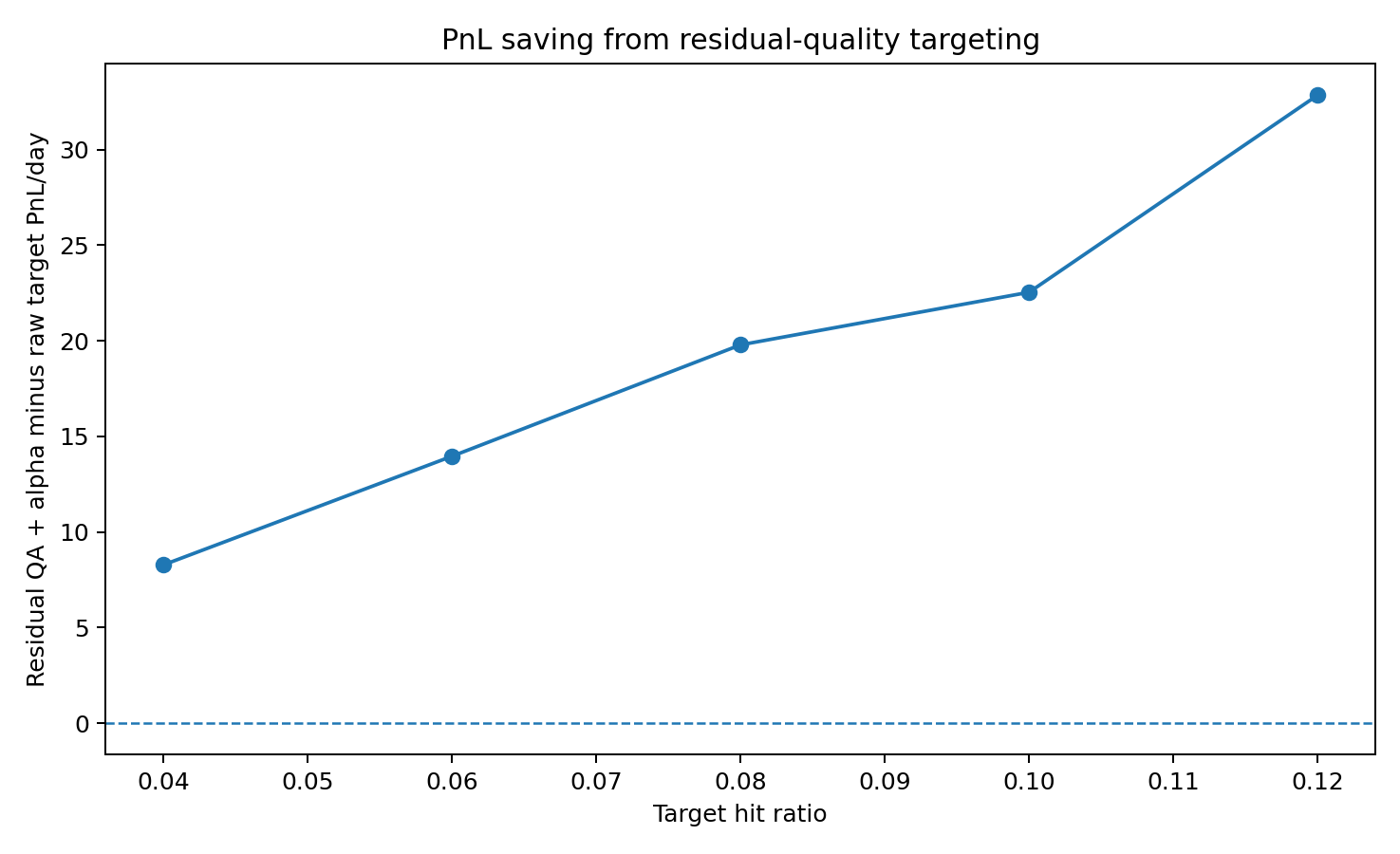}
\caption{PnL saving from residual-quality targeting over raw targeting by hit-ratio target in the robustness run.}
\label{fig:app_step4_saving_target}
\end{figure}

\begin{figure}[H]
\centering
\includegraphics[width=0.76\textwidth]{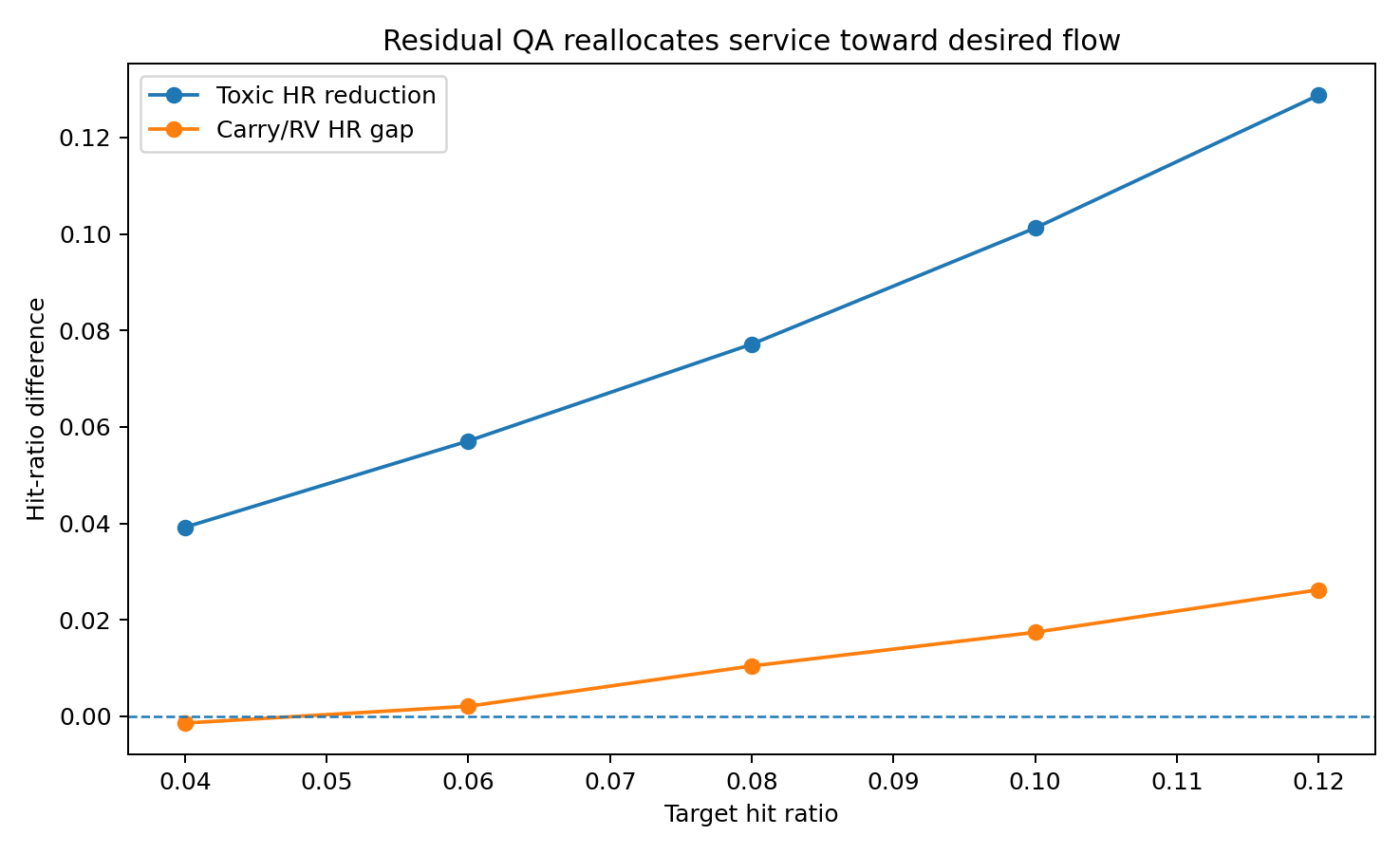}
\caption{Service reallocation by target: residual-quality targeting reduces residual-toxic hit ratio and preserves service to better flow.}
\label{fig:app_step4_reallocation}
\end{figure}

\begin{figure}[H]
\centering
\includegraphics[width=0.76\textwidth]{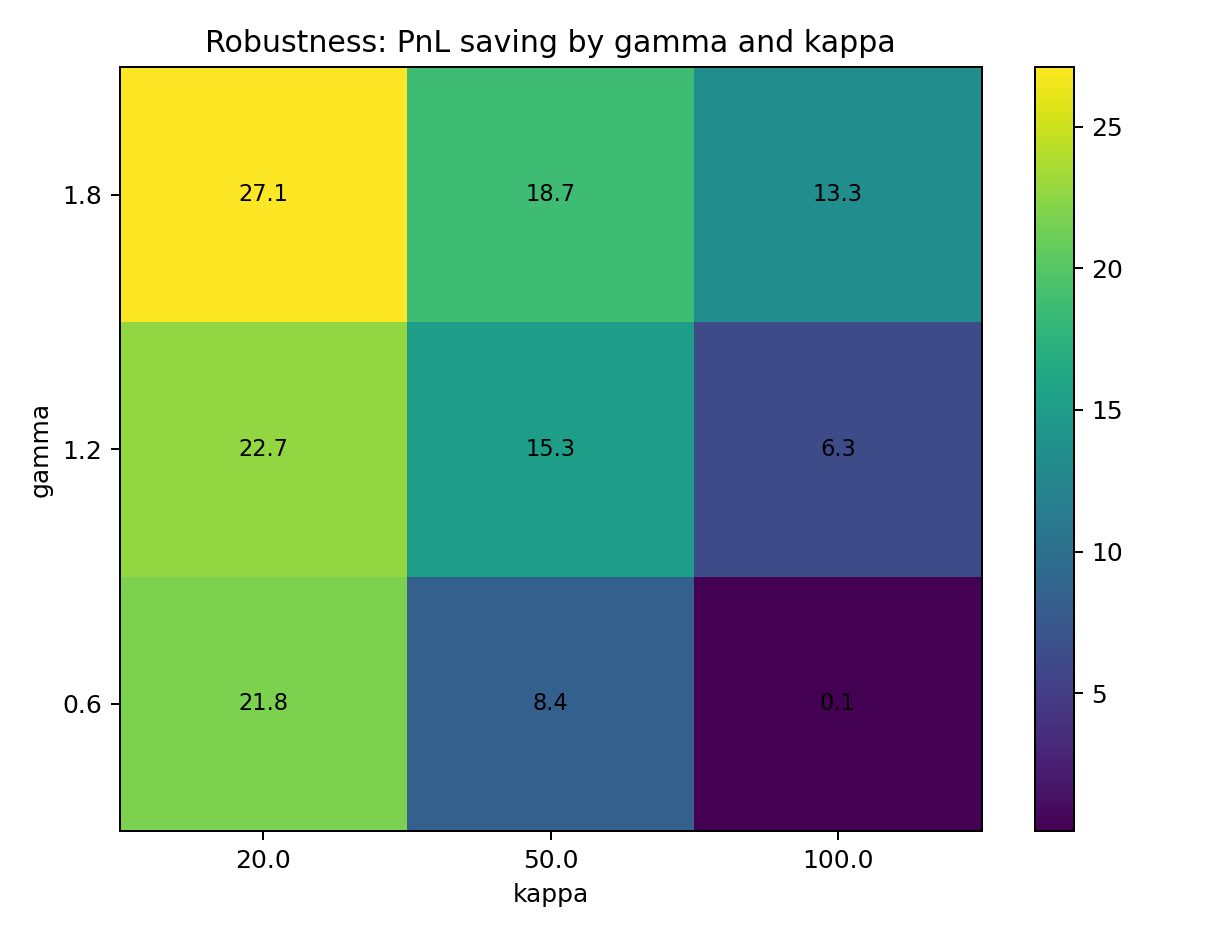}
\caption{Robustness heatmap: PnL saving as a function of quality-weight strength $\gamma$ and hit-ratio penalty strength $\kappa$. This sweep uses a local dual approximation and is reported as supplementary evidence; the main text uses the nonlinear-dual solve.}
\label{fig:app_step4_gk_pnl}
\end{figure}

\begin{figure}[H]
\centering
\includegraphics[width=0.76\textwidth]{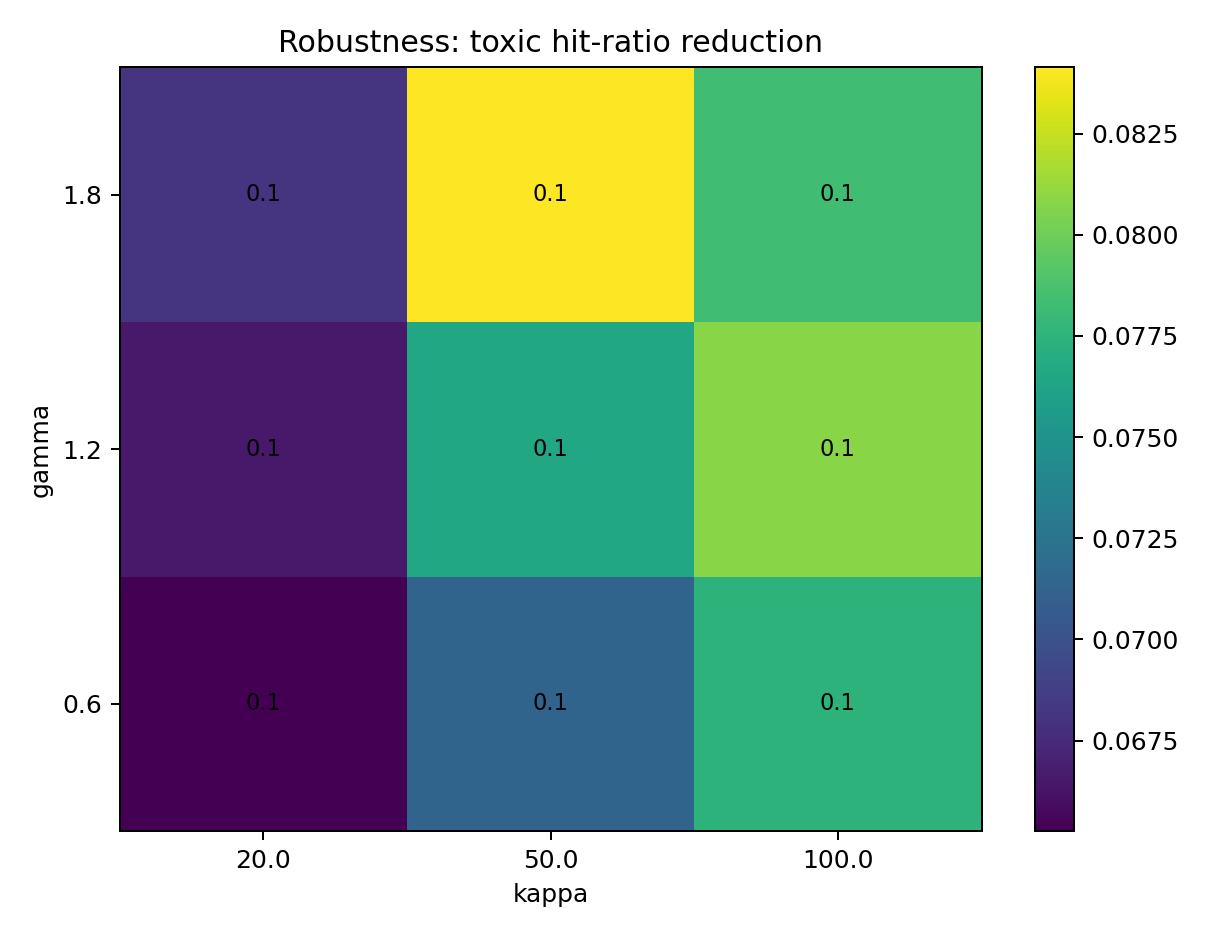}
\caption{Robustness heatmap: residual-toxic hit-ratio reduction as a function of $\gamma$ and $\kappa$.}
\label{fig:app_step4_gk_toxic}
\end{figure}

\begin{figure}[H]
\centering
\includegraphics[width=0.76\textwidth]{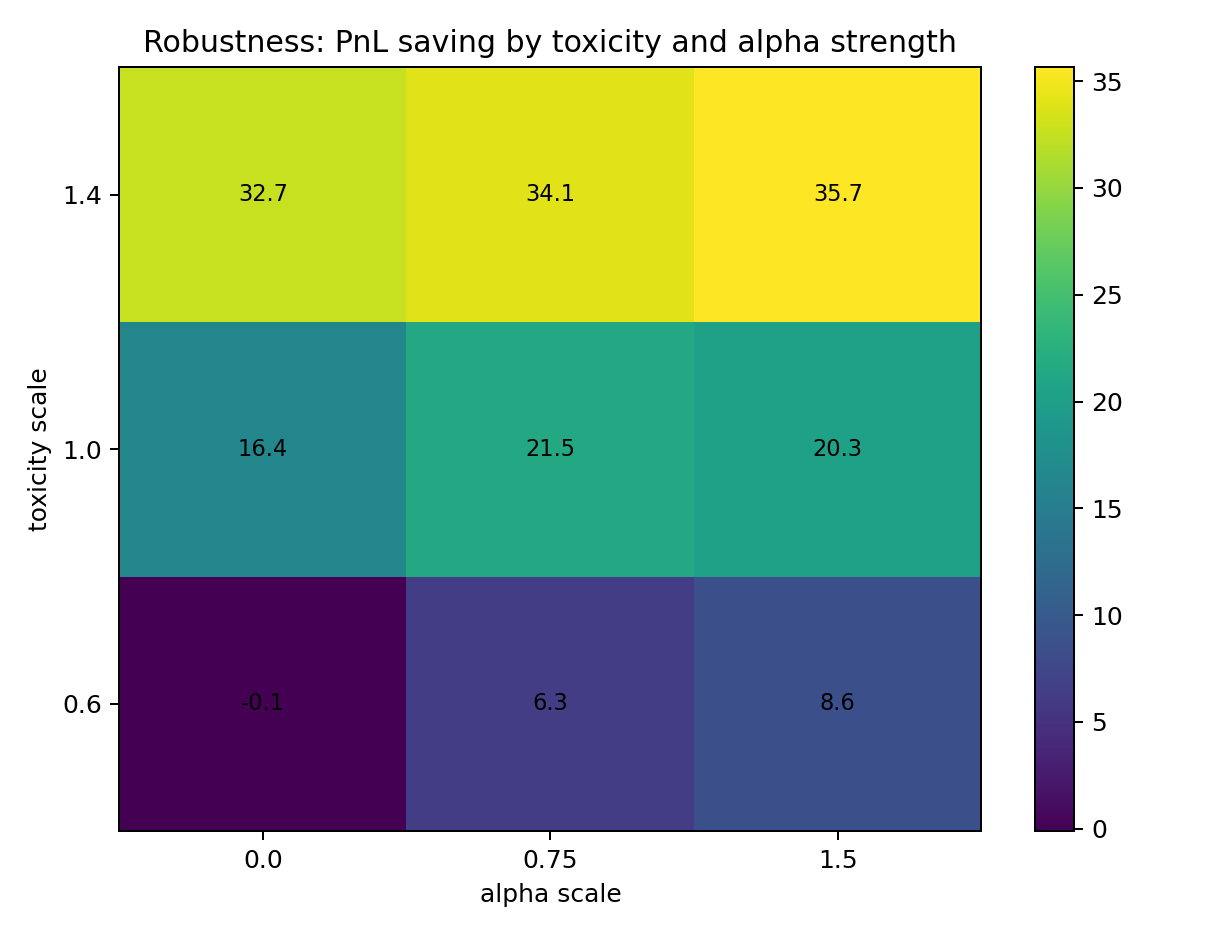}
\caption{Robustness heatmap: PnL saving as a function of residual-toxicity scale and credit-alpha scale.}
\label{fig:app_step4_tox_alpha}
\end{figure}

\begin{figure}[H]
\centering
\includegraphics[width=0.76\textwidth]{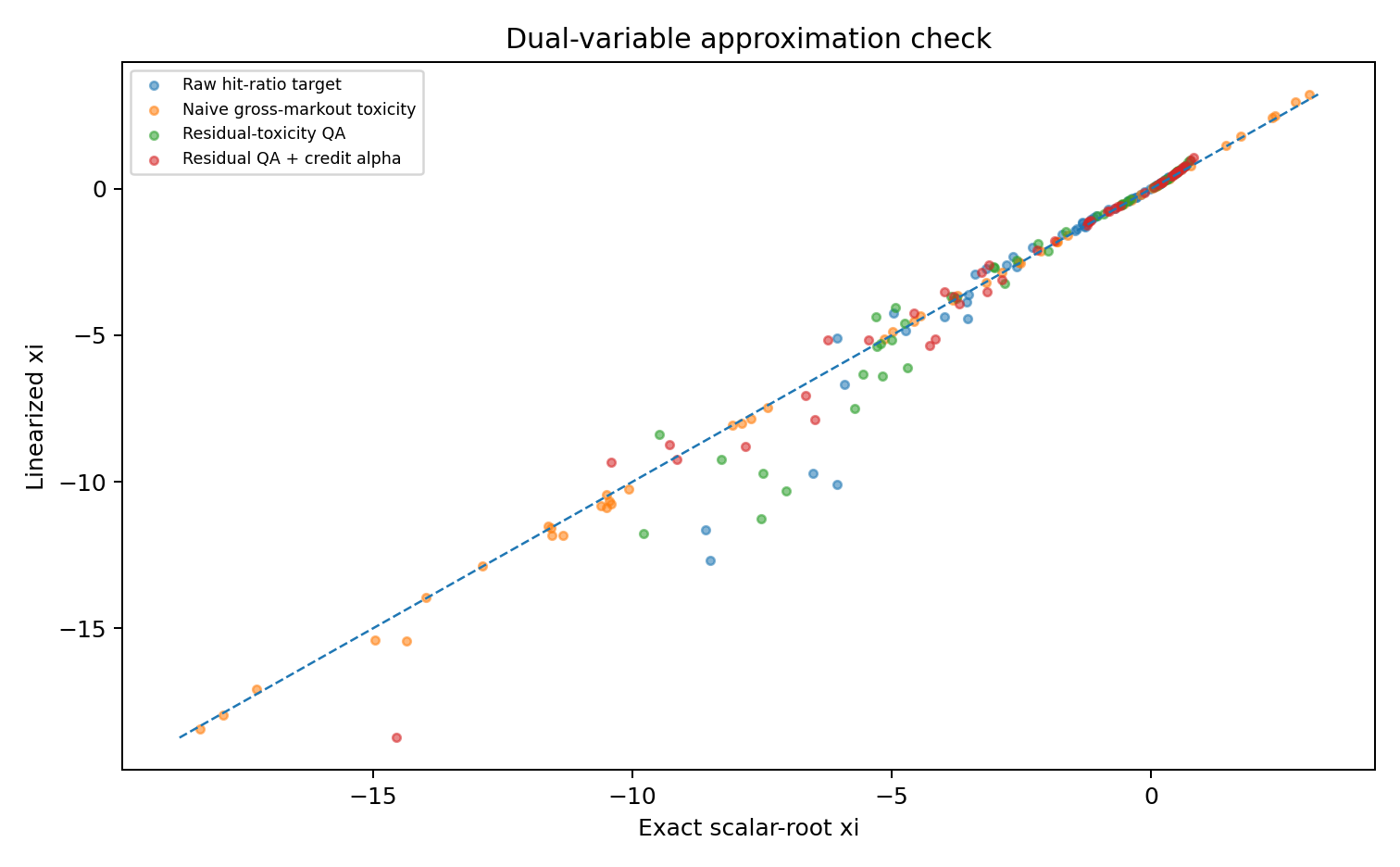}
\caption{Exact scalar dual variable versus the local linearized dual approximation. The comparison motivates using the nonlinear solve for headline results.}
\label{fig:app_step4_exact_linear_xi}
\end{figure}

\begin{figure}[H]
\centering
\includegraphics[width=0.76\textwidth]{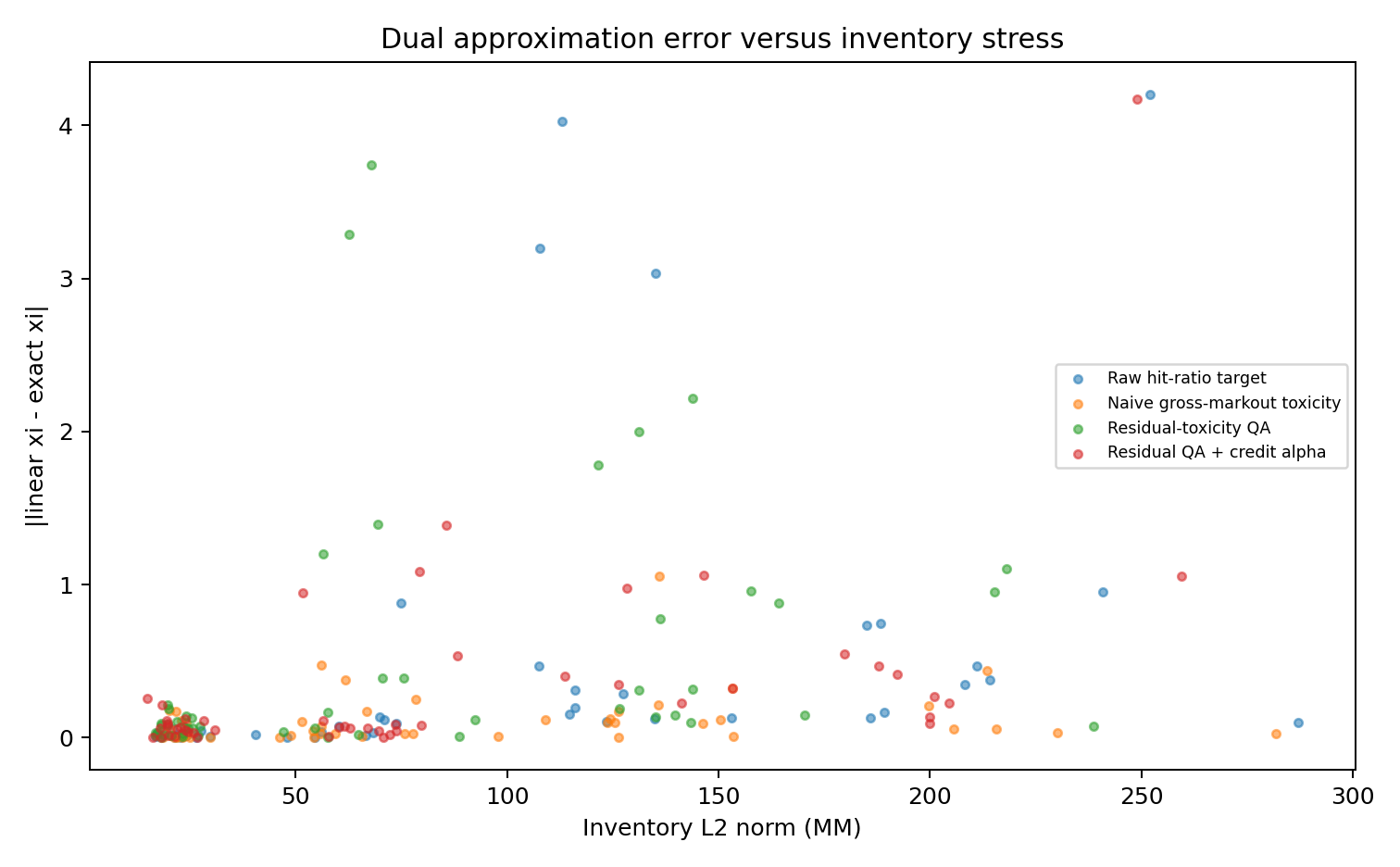}
\caption{Error of the linearized dual approximation versus inventory stress. Approximation errors are larger in stressed states, so the headline results use the nonlinear scalar solve.}
\label{fig:app_step4_xi_error}
\end{figure}

\subsection{Passive/index extension supplementary diagnostics}

\begin{figure}[H]
\centering
\includegraphics[width=0.76\textwidth]{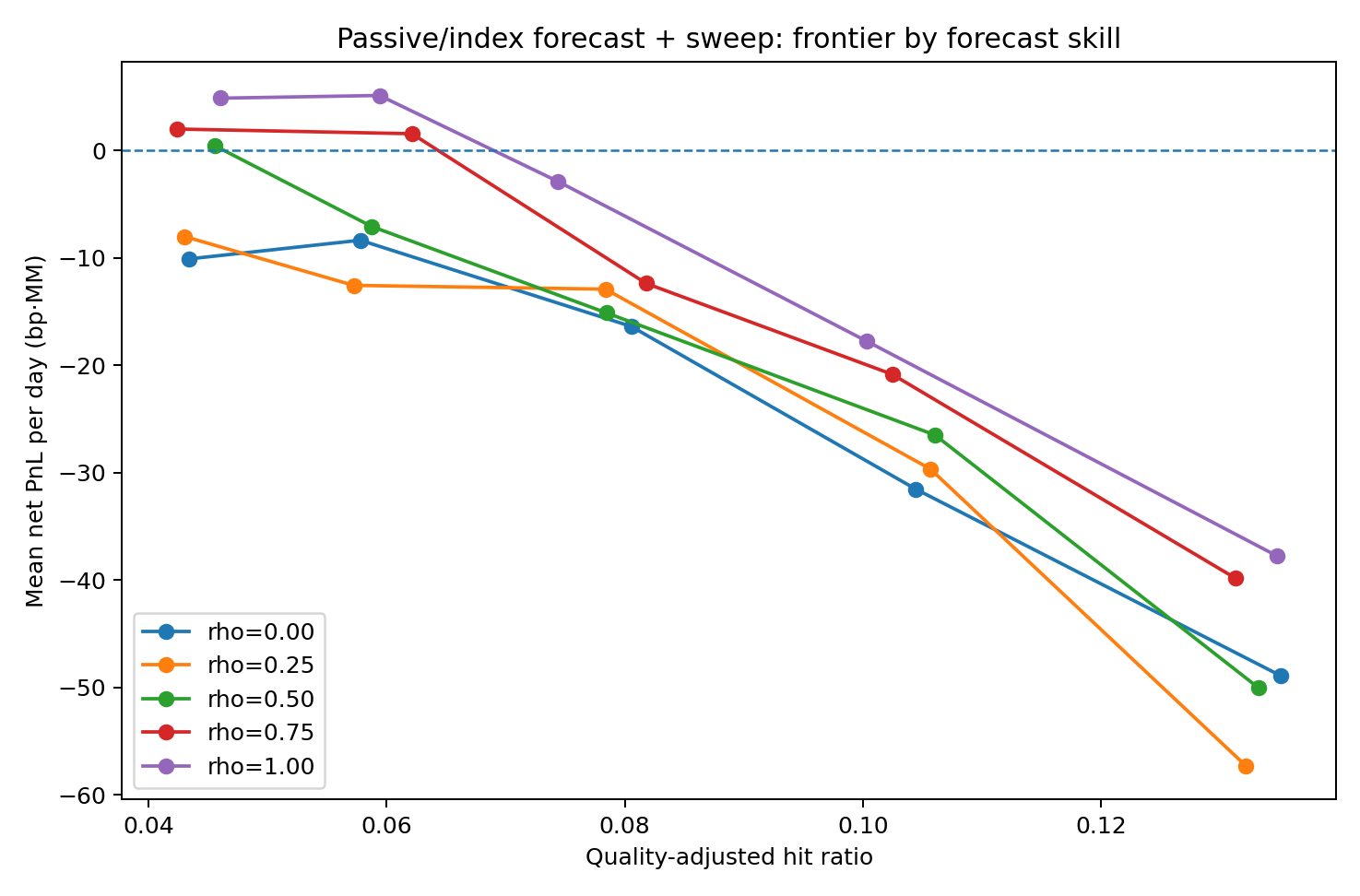}
\caption{Passive/index extension frontier by forecast skill with a fixed 50 percent sweep pre-positioning fraction. The figure shows that aggressive pre-positioning is not uniformly beneficial.}
\label{fig:app_step5_frontier_skill}
\end{figure}

\begin{figure}[H]
\centering
\includegraphics[width=0.76\textwidth]{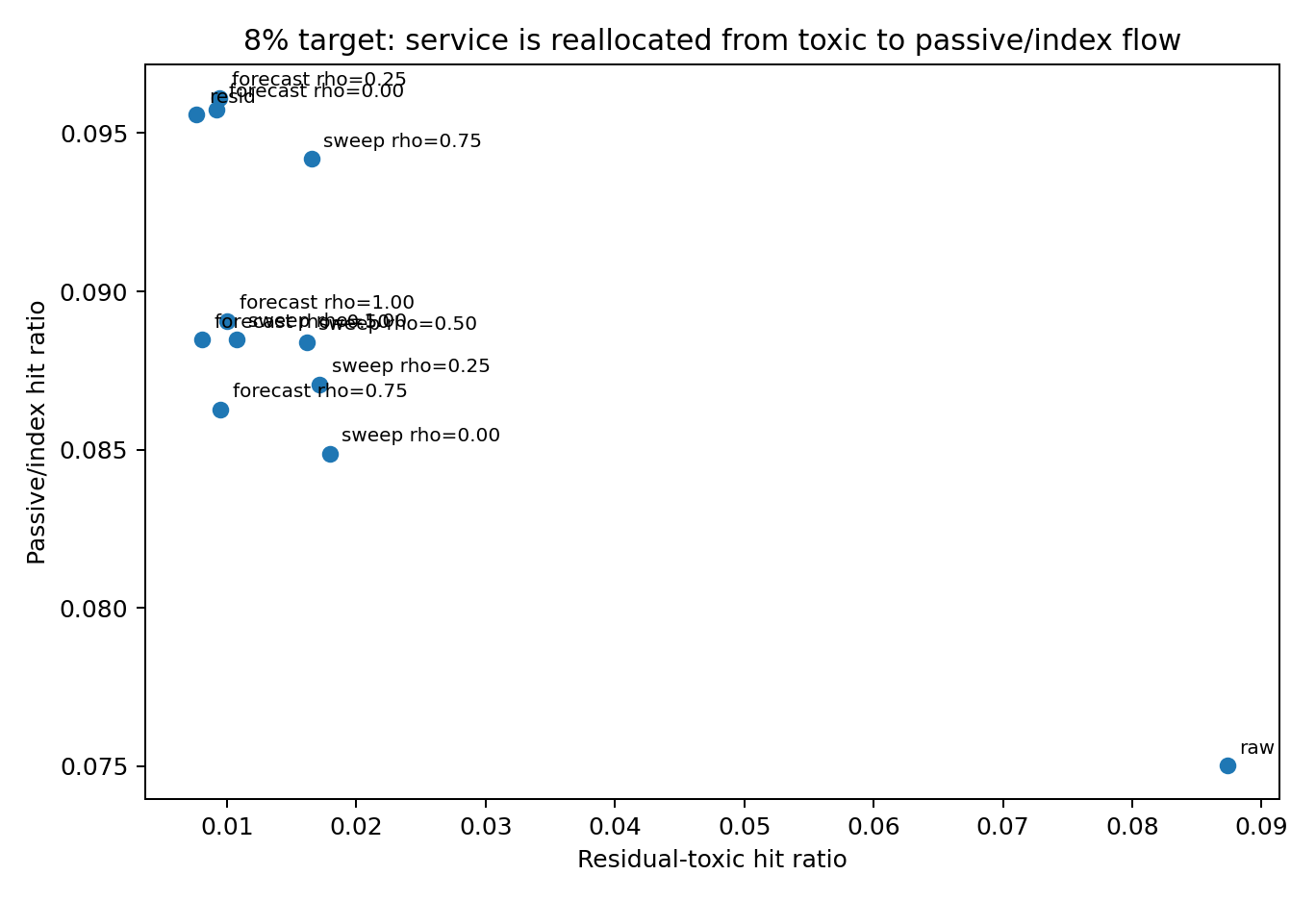}
\caption{Passive/index versus residual-toxic hit ratios at the 8 percent target. The extension reallocates service toward low-residual-toxicity passive/index flow.}
\label{fig:app_step5_passive_toxic}
\end{figure}

\begin{figure}[H]
\centering
\includegraphics[width=0.76\textwidth]{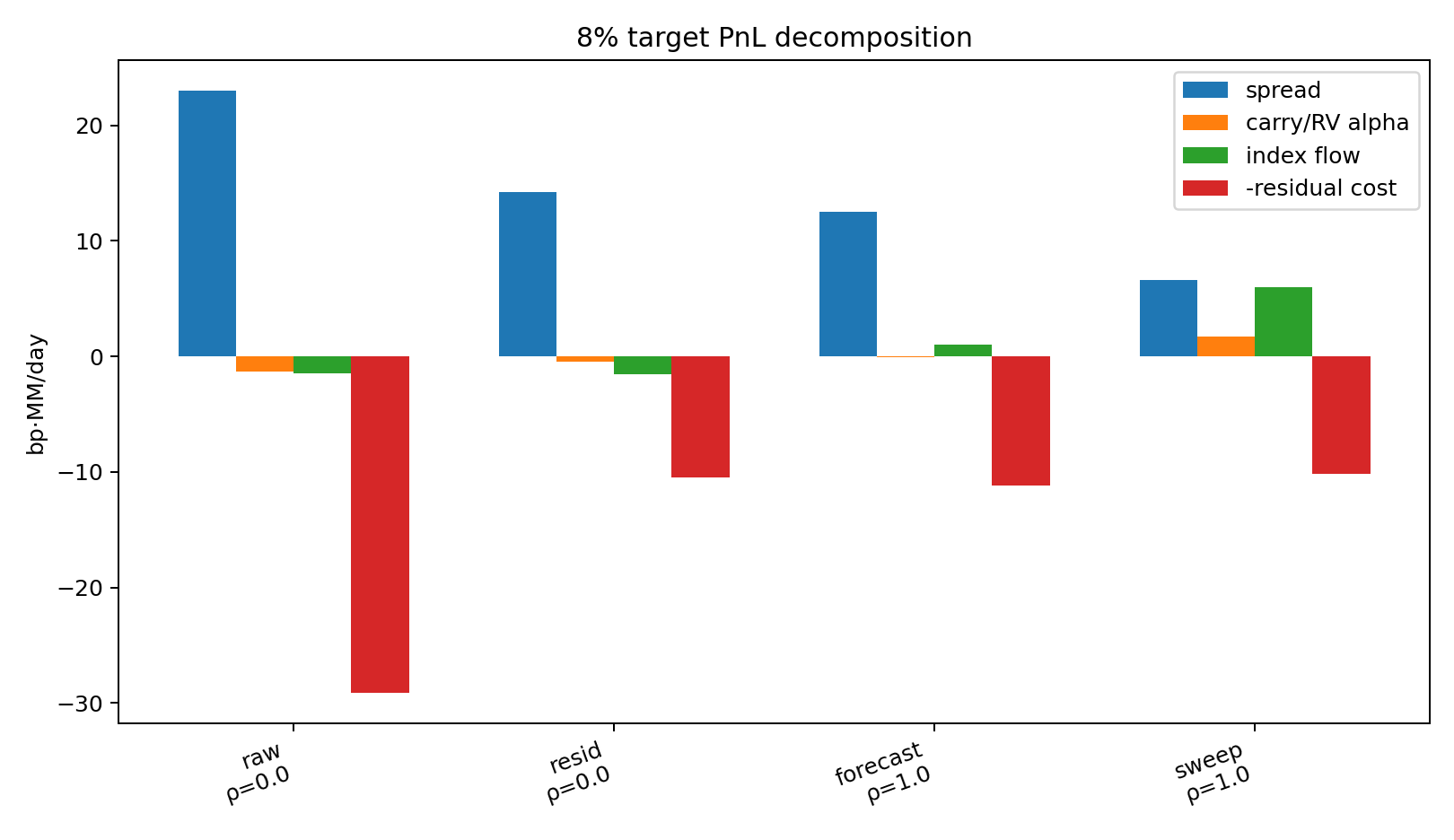}
\caption{PnL decomposition for the passive/index extension. Forecastable index flow adds recycling value but residual costs and liquidation costs remain important.}
\label{fig:app_step5_pnl_decomp}
\end{figure}

\begin{figure}[H]
\centering
\includegraphics[width=0.76\textwidth]{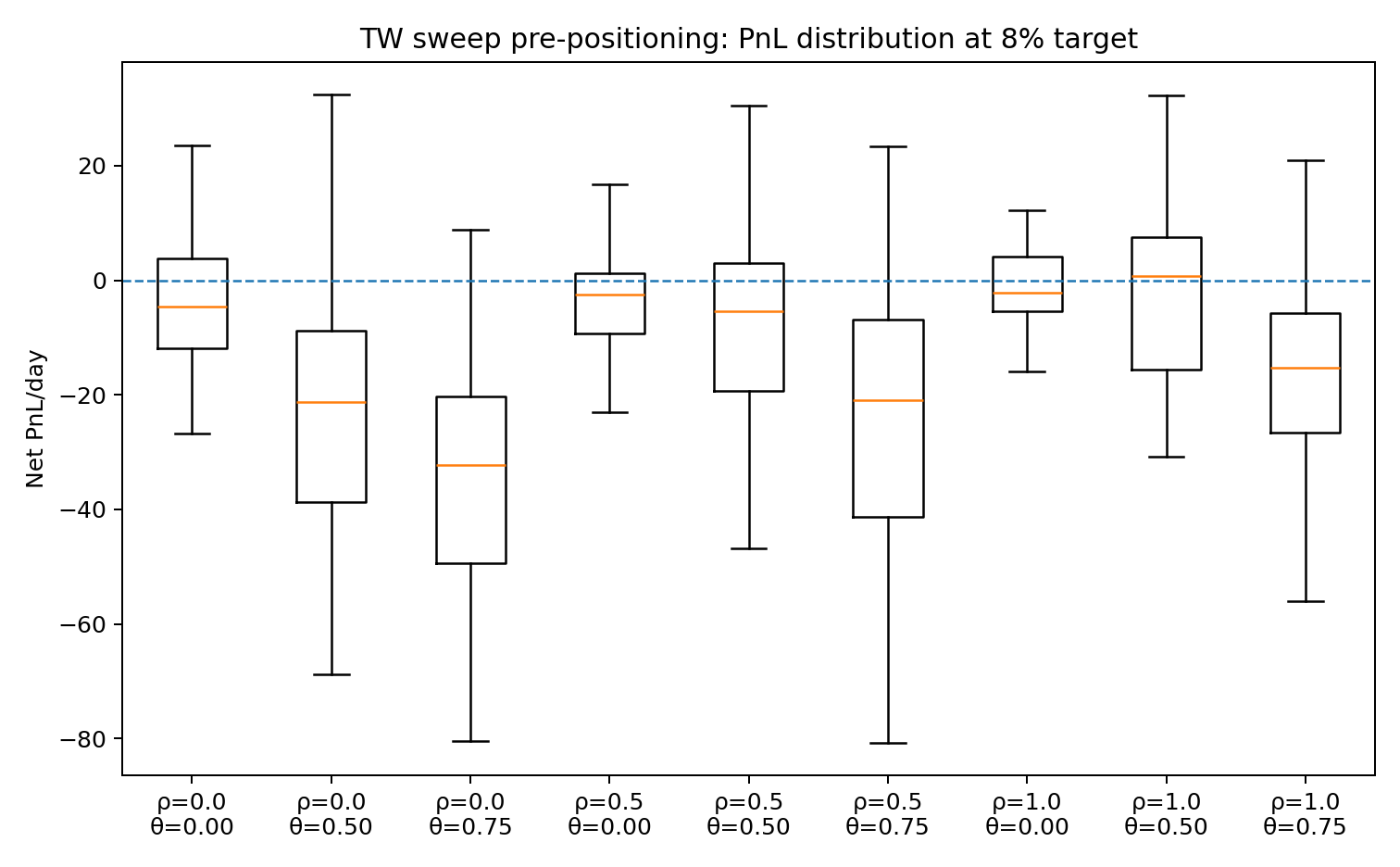}
\caption{TW Sweep pre-positioning PnL distribution at the 8 percent target across selected forecast-skill and sweep-fraction scenarios.}
\label{fig:app_step5_sweep_dist}
\end{figure}

\subsection{Additional nonlinear-dual implementation checks}

\begin{figure}[H]
\centering
\includegraphics[width=0.76\textwidth]{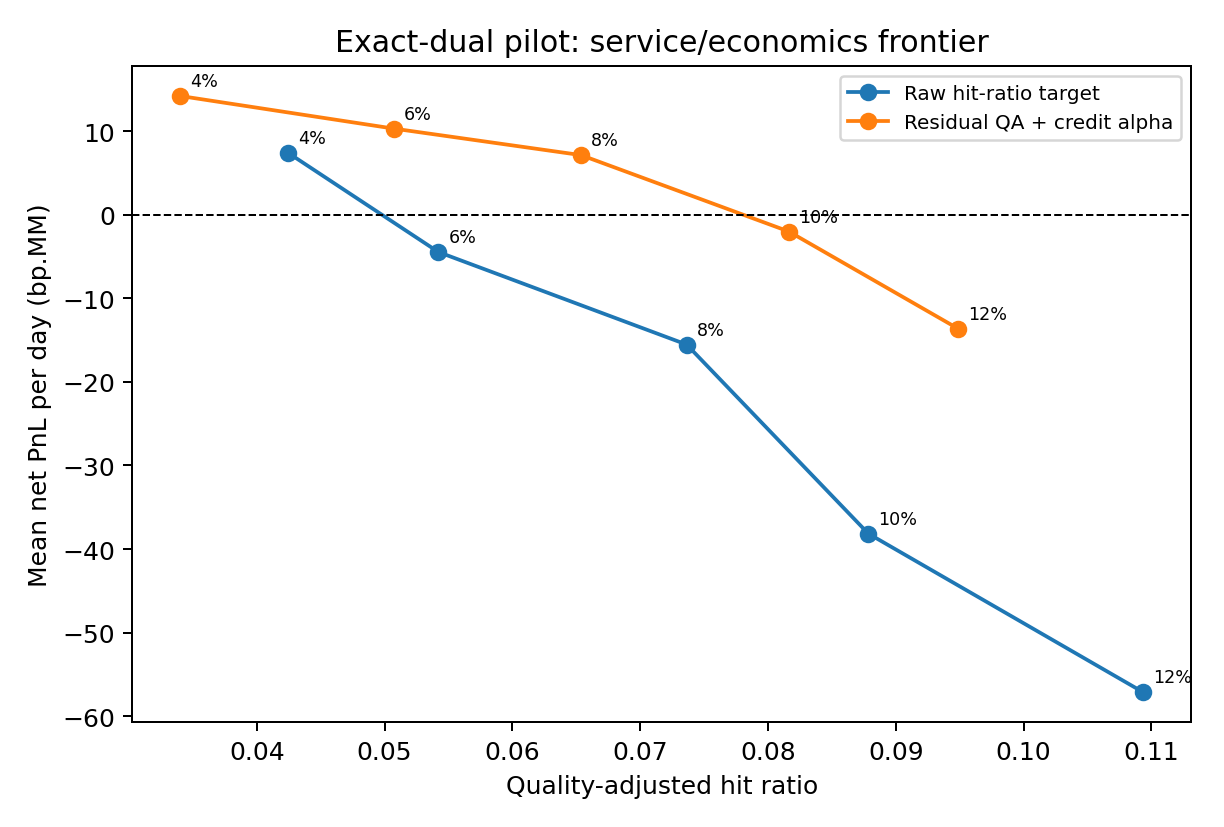}
\caption{Small-sample nonlinear-dual frontier used as an implementation check before the larger attained-quality experiment.}
\label{fig:app_pilot_frontier}
\end{figure}

\begin{figure}[H]
\centering
\includegraphics[width=0.76\textwidth]{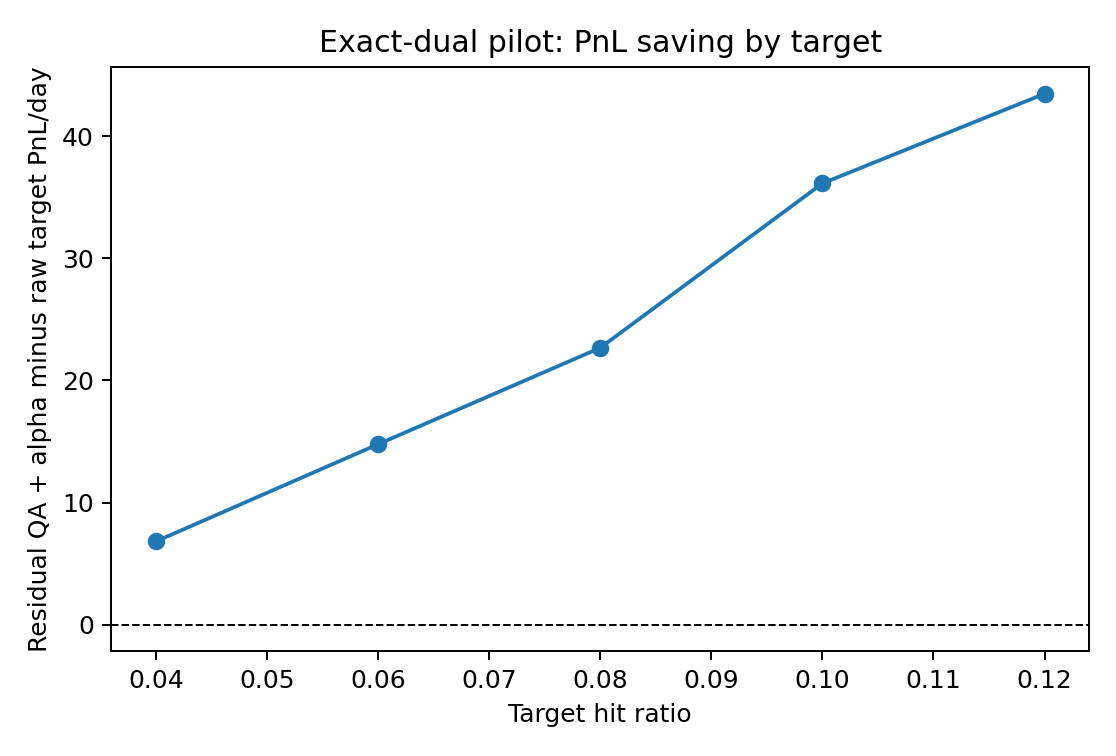}
\caption{Small-sample nonlinear-dual PnL-saving check. The main text reports the larger attained-quality frontier.}
\label{fig:app_pilot_saving}
\end{figure}

\begin{figure}[H]
\centering
\includegraphics[width=0.76\textwidth]{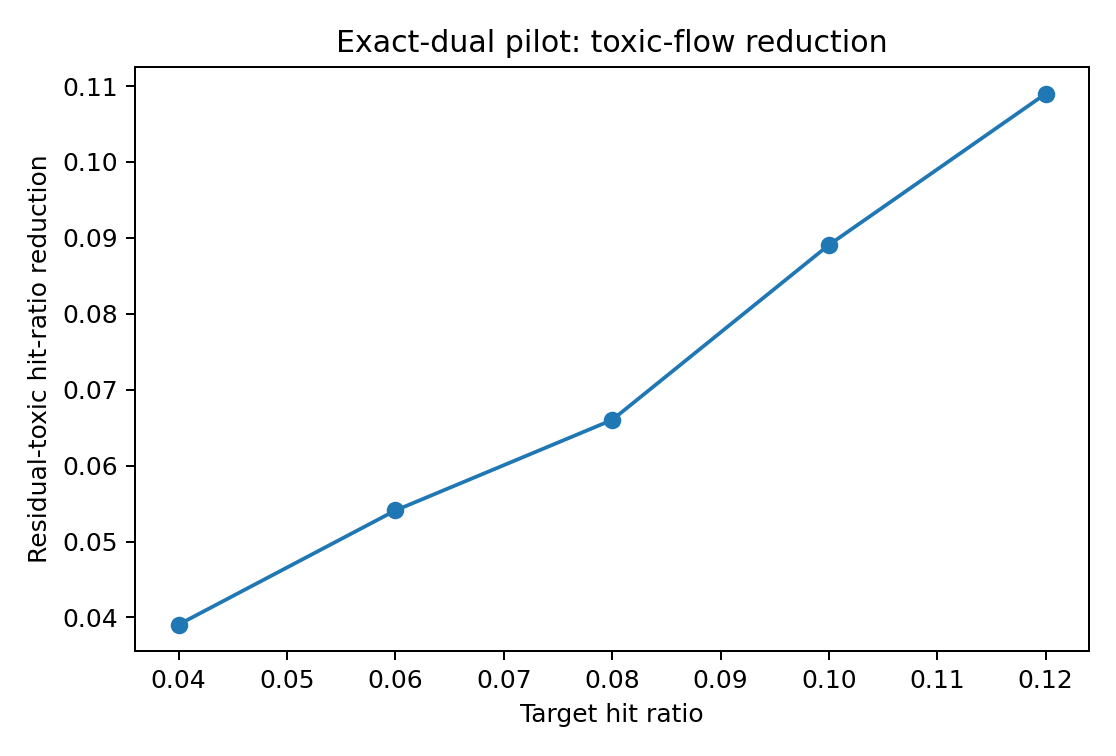}
\caption{Small-sample nonlinear-dual toxic-flow reduction check. The main text reports the larger flow-allocation analysis.}
\label{fig:app_pilot_toxic}
\end{figure}

\subsection{Risk-aware style-flow warehousing supplementary diagnostics}

\begin{figure}[H]
\centering
\includegraphics[width=0.76\textwidth]{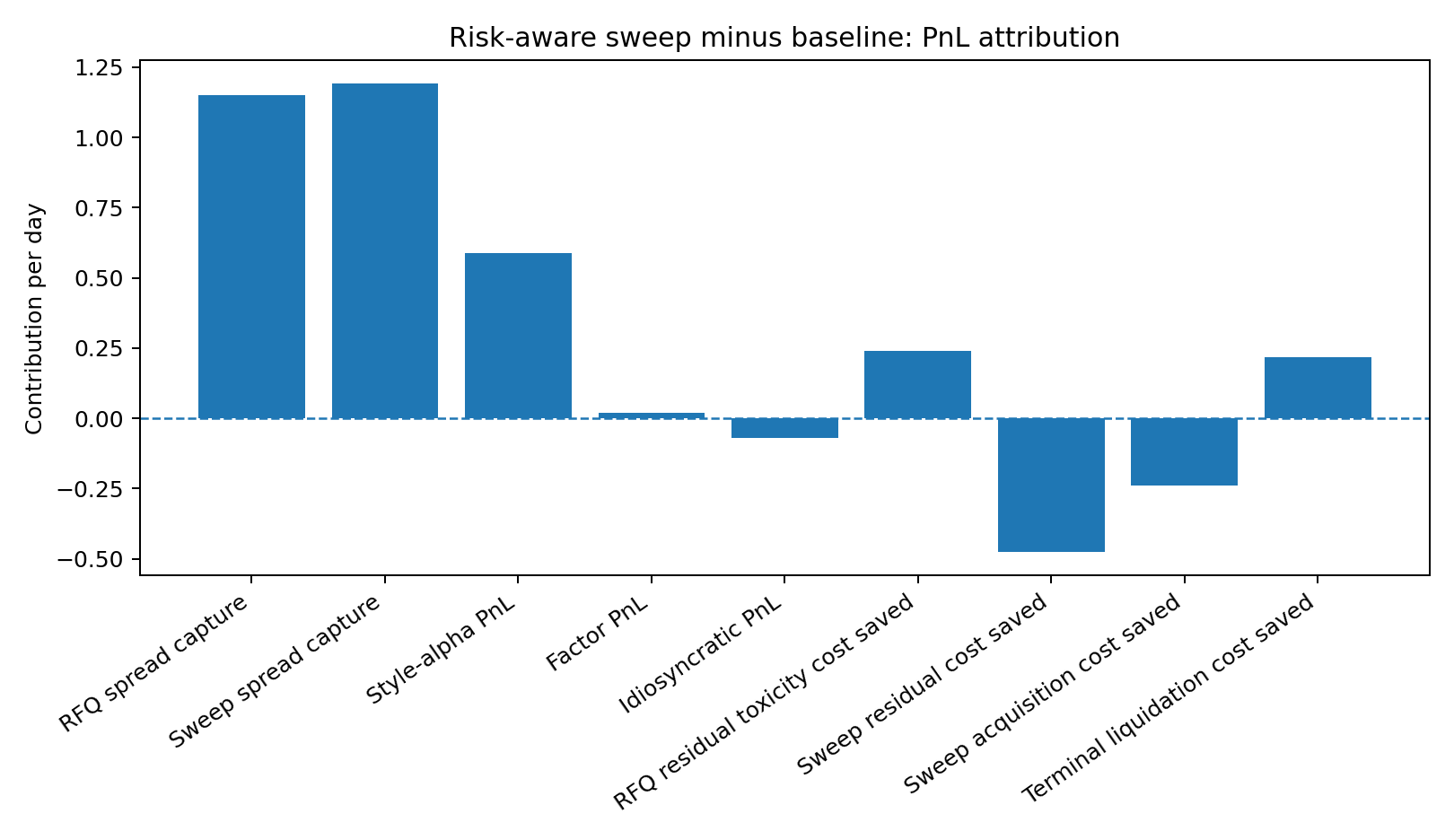}
\caption{Risk-aware style-flow PnL attribution relative to the residual-quality baseline. The attribution shows the trade-off between Sweep spread capture, style-alpha PnL, Sweep costs, and liquidation-cost savings.}
\label{fig:app_risk_aware_attr}
\end{figure}

\begin{figure}[H]
\centering
\includegraphics[width=0.76\textwidth]{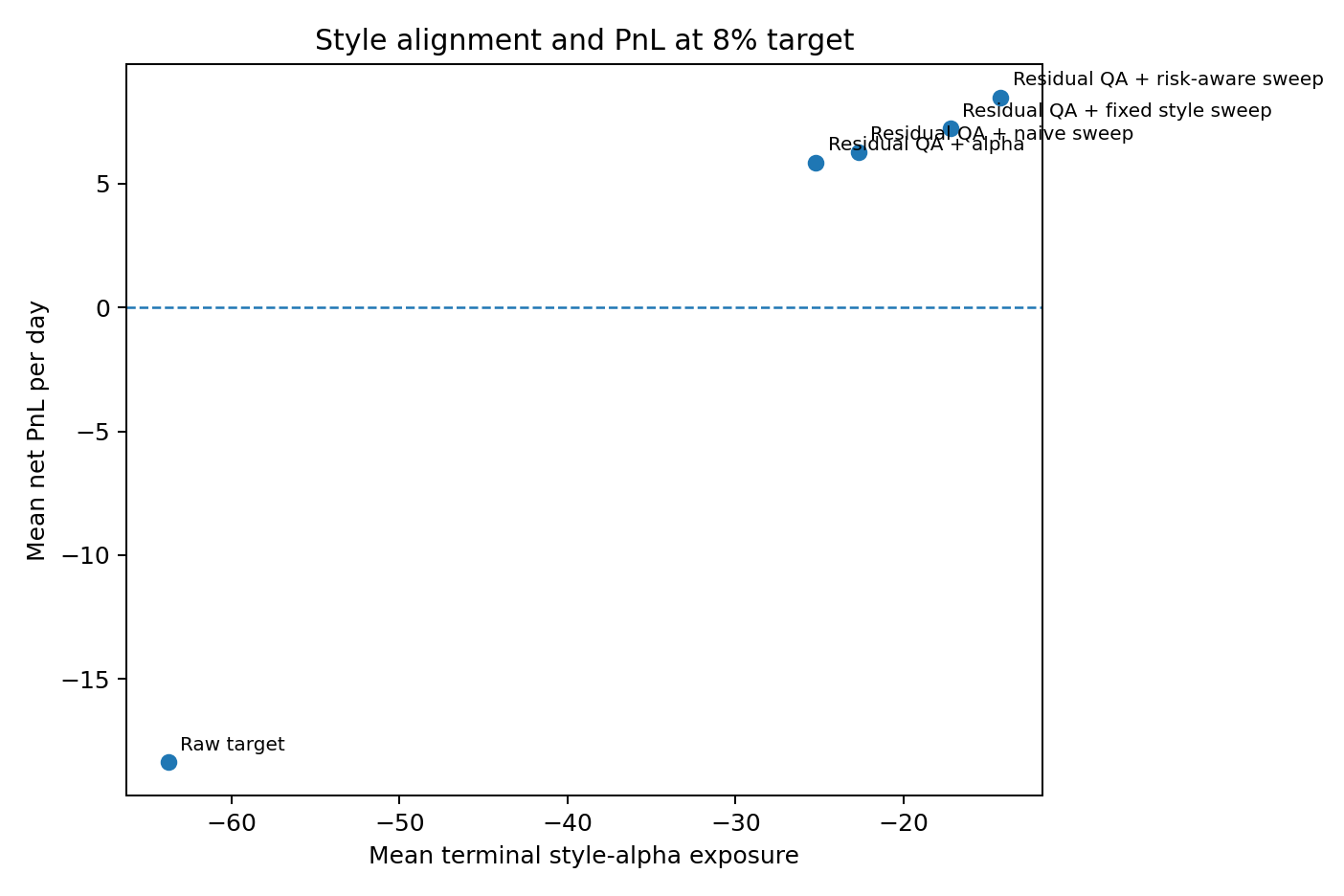}
\caption{Style-alpha exposure and PnL at the 8 percent target. Risk-aware participation improves the quality of acquired inventory while keeping inventory risk controlled.}
\label{fig:app_risk_aware_alignment}
\end{figure}

\begin{figure}[H]
\centering
\includegraphics[width=0.76\textwidth]{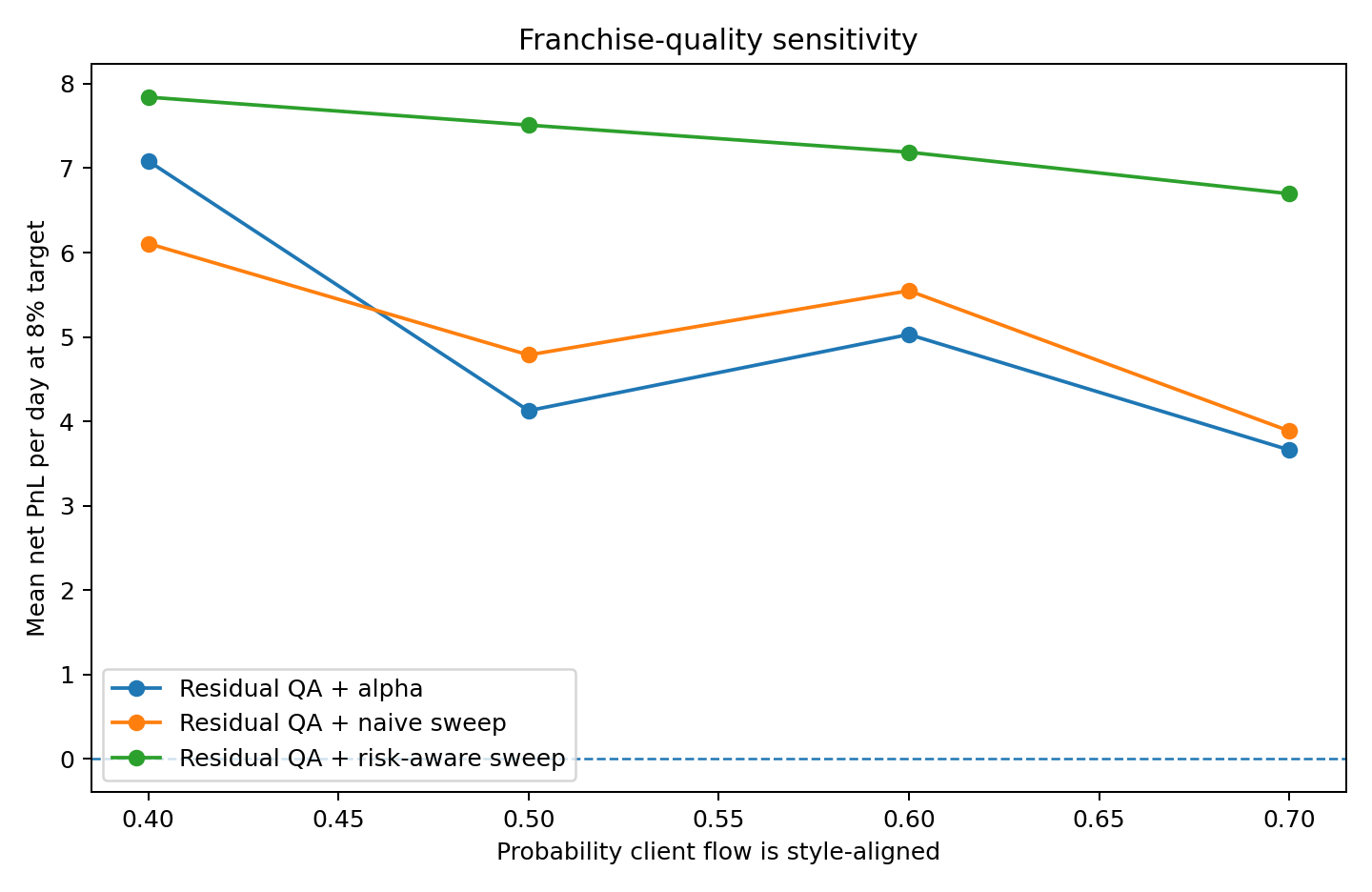}
\caption{Franchise-quality sensitivity. The probability that client inventory is style-aligned affects the economics of Sweep participation. The risk-aware rule remains above the no-Sweep and naive-Sweep baselines in this synthetic experiment.}
\label{fig:app_risk_aware_franchise}
\end{figure}

\end{document}